\documentclass[preprint,12pt]{elsarticle}

\usepackage{bm}
\usepackage{comment}
\def\tsc#1{\csdef{#1}{\textsc{\lowercase{#1}}\xspace}}
\tsc{WGM}
\tsc{QE}
\tsc{EP}
\tsc{PMS}
\tsc{BEC}
\tsc{DE}

\usepackage{framed} 
\usepackage{multicol}
\usepackage{nomencl} 
\makenomenclature
\renewcommand*\nompreamble{\begin{multicols}{2}}
\renewcommand*\nompostamble{\end{multicols}}

\usepackage{algorithm}
\usepackage{algpseudocode}
\usepackage{amsmath,amsfonts,amsthm}
\usepackage{graphicx}
\usepackage{url}
\usepackage[table,dvipsnames]{xcolor}
\usepackage{hyperref} 
\usepackage{booktabs}
\usepackage{makecell}
\usepackage{amsmath}
\usepackage{mhchem}
\usepackage{lineno}
\usepackage{fancyhdr}
\usepackage{adjustbox}
\usepackage{amssymb}
\usepackage{float}
\usepackage{amsmath,amssymb,amsthm}
\theoremstyle{plain}
\newtheorem{lemma}{Lemma}[section]

\algrenewcommand\algorithmicrequire{\textbf{Input:}}
\algrenewcommand\algorithmicensure{\textbf{Output:}}
\usepackage[caption=false,font=normalsize,labelfont=sf,textfont=sf]{subfig}
\hypersetup{
    colorlinks=true, 
    linkcolor=black, 
    filecolor=black, 
    urlcolor=black, 
    citecolor=black, 
    linkbordercolor={1 1 1}, 
    pdfborder={0 0 0} 
}

\journal{Applied Energy}

\begin{document}
\begin{frontmatter}



\title{Optimal sizing of a hydrogen-based direct reduced iron–electric arc furnace system integrated with methanol synthesis toward zero-carbon steel production}

\affiliation[1]{organization={Department of Electrical Engineering, Tsinghua University},
                city={Beijing},
                postcode={100084}, 
                country={China}}
\affiliation[2]{organization={School of Modern Chemical Engineering, Ningxia Institute of Applied Technology},
                city={Ningxia},
                postcode={753000}, 
                country={China}}
\affiliation[3]{organization={State Key Laboratory of Power System Operation and Control},
                city={Beijing},
                postcode={100084}, 
                country={China}}
\affiliation[4]{organization={State Key Laboratory of Smart Power Distribution Equipment and System, State Grid Jibei Electric Power Company Limited},
                city={Beijing},
                postcode={100083}, 
                country={China}}
\author[1]{Qiang Ji}

\author[1]{Fashun Shi \corref{cor1}}
\ead{shi_fashun@tsinghua.edu.cn}
\cortext[cor1]{Corresponding author}

\author[1]{Lin Cheng}

\author[2]{Yeye Xie}

\author[3]{Yufei Xi}

\author[1]{Zhen Dai}

\author[4]{Xun Wang}

\begin{abstract}
A hydrogen-based direct reduced iron-electric arc furnace system integrated with methanol synthesis (\ce{H2}-DRI-EAF-MeOH) provides a feasible pathway toward zero-carbon steel production. However, the volatility and intermittency of renewable energy sources (RES) may lead to capacity oversizing, while conventional annualized-cost and levelized-cost metrics are insufficient to evaluate investment return. To address these issues, this paper develops a fractional programming-based optimal sizing model for the \ce{H2}-DRI-EAF-MeOH system. Hourly production rate limits are decoupled from annual production capacities and formulated as independent decision variables. Theoretical analysis demonstrates that increasing hourly production rate limits expands the multi-period operating feasible region, increases the theoretical upper bound on RES matching, and provides an investment-efficiency condition for increasing production rate limits. The annualized net return on investment (ANROI) is then adopted to represent investment return, and the resulting mixed-integer linear fractional programming model is solved using the Dinkelbach method. Case studies show that compared with the annualized net return (ANR) benchmark, the ANROI-based sizing yields higher internal rates of return (IRR), reaching 20.67\% and 17.08\% under the on-grid and off-grid scenarios, respectively. As the equivalent annual operating hours decrease from 8000 to 4000 h, the IRR increases from 19.68\% to 20.67\% under the on-grid scenario and from 12.64\% to 17.08\% under the off-grid scenario, while the corresponding battery storage capacities decrease by 113.56 and 796.31 MWh. The off-grid zero-carbon constraint enables \ce{CO2}-to-MeOH conversion but reduces the system IRR to 14.99\%, indicating that the resulting revenues are insufficient to offset the additional costs.
\end{abstract}
\begin{graphicalabstract}
\begin{adjustbox}{width=1\textwidth,center}
\includegraphics{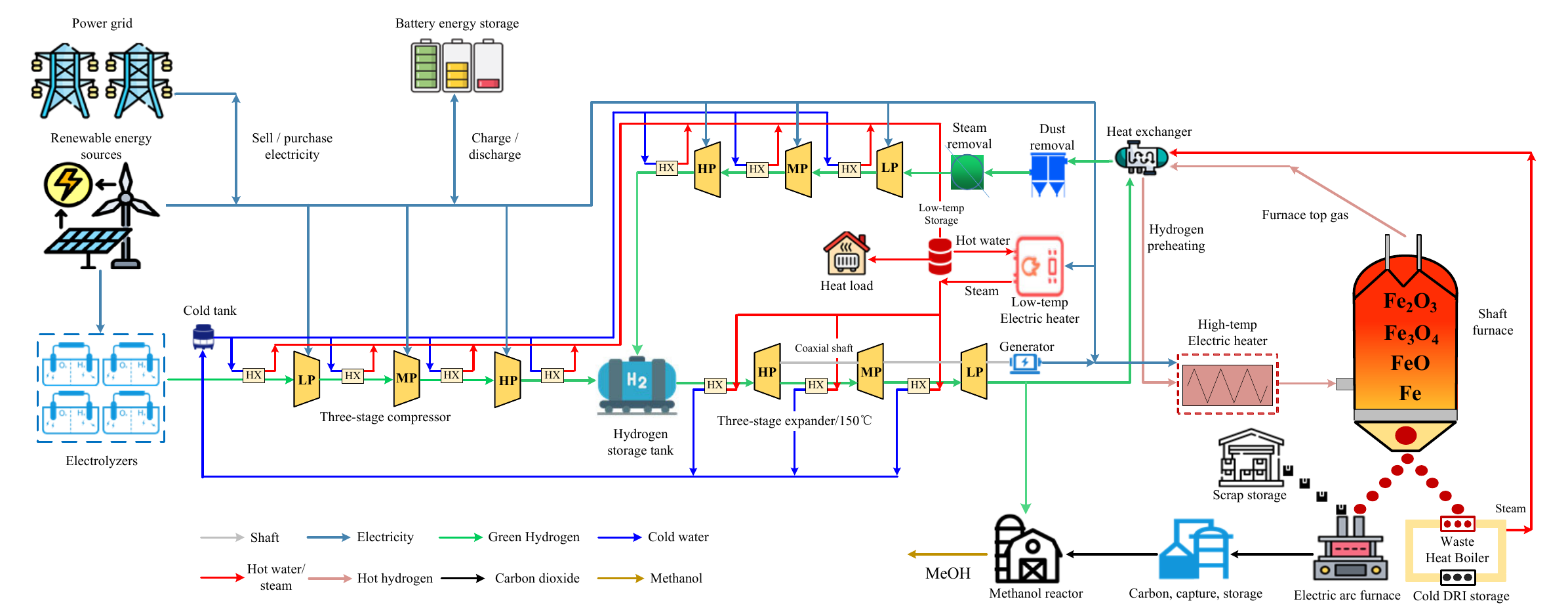}
\end{adjustbox}
\end{graphicalabstract}

\begin{highlights}
    \item On-grid and off-grid IRRs reach 20.67\% and 17.08\%, respectively.
    \item Hourly production rate improves return and reduces battery storage.
    \item Off-grid zero-carbon constraint lowers the system IRR to 14.99\%.
\end{highlights}
\end{frontmatter}
\section{Introduction}
\label{sec1}
\subsection{Background}
Accounting for approximately 7\% of global \ce{CO2} emissions and 8\% of global energy consumption \cite{lei}, the iron and steel industry is a key sector for deep industrial decarbonization and climate-change mitigation \cite{meng2025technologies,su2023multi}. However, achieving deep decarbonization in this sector remains challenging \cite{xu2023plant,ji2025energy}, as the blast furnace-basic oxygen furnace (BF-BOF) route still accounts for about 71.2\% of global crude steel production \cite{ariyama2019diversification}. This route relies heavily on coke and pulverized coal \cite{gu2024co2}, emitting approximately 1.9 t \ce{CO2} per tonne of crude steel \cite{sheng2024rational}. Although retrofitting measures such as waste heat recovery \cite{zhou2024application}, hydrogen-enriched carbon recycling \cite{xia2026study}, and top gas recycling \cite{perpinan2023power} can reduce \ce{CO2} emissions from the BF-BOF route, their mitigation potential remains limited because the route still relies on carbon-based reductants. To overcome the carbon-based reduction bottleneck of the BF-BOF route, a hydrogen-based direct reduced iron-electric arc furnace system integrated with methanol synthesis (\ce{H2}-DRI-EAF-MeOH) provides a feasible pathway toward zero-carbon steel production. In this system, hydrogen replaces carbon-based reductants in iron ore reduction (\ce{Fe2O3 + 3H2 -> 2Fe + 3H2O}), while captured \ce{CO2} is converted into MeOH (\ce{CO2 + 3H2 -> CH3OH + H2O}). \par
For the RES-driven \ce{H2}-DRI-EAF-MeOH system, the intermittency and variability of RES \cite{ostergaard2022renewable} cause temporal mismatches between renewable electricity supply and the power demand of the electrolyzer, shaft furnace (SF), electric arc furnace (EAF), and methanol synthesis reactor (MSR). Battery energy storage (BS) is commonly used to buffer these mismatches \cite{saini2022investigation,yang2022modelling,levin2023energy}; however, excessive BS expansion may increase capital investment and reduce investment return. 
Accordingly, it is necessary to decouple maximum hourly production rates from annual production capacities and treat it as an independent decision variable, thereby mitigating temporal mismatches through production rate flexibility and reducing reliance on BS capacity expansion. However, because different levels of production rate flexibility entail different capital investments, maximizing annualized net return (ANR) may favor solutions with higher absolute returns but lower investment efficiency. Therefore, annualized net return on investment (ANROI) is adopted as the fractional objective to account for investment efficiency in optimal sizing. \par 
\subsection{Literature review}
Existing studies \cite{shahabuddin2023decarbonisation,rumsa2025global} on the optimal sizing of low-carbon steelmaking systems have primarily focused on multi-stage capacity planning. For example, Ref.~\cite{sheng2024rational} proposed a multi-stage expansion planning model with dynamic techno-economic parameters and adopted a rolling-horizon method to optimize long-term capacity investment strategies for renewable hydrogen-based steelmaking systems. Ref.~\cite{wu2025multi} developed a multi-objective, multi-stage capacity planning model for a wind-natural gas-hydrogen coupled steel production system with a stage-adjustable \ce{H2}-CO ratio. However, these studies paid limited attention to the effects of the operational flexibility of key production units on the levelized cost of steel production. In contrast, other techno-economic studies have analyzed operational flexibility and the oversizing of key production units in low-carbon steelmaking systems \cite{salimbeni2023techno,boldrini2024flexibility,weiss2024flexible}. Ref.~\cite{bhaskar2022decarbonizing} developed a techno-economic model for an on-grid \ce{H2}-DRI-EAF system to quantify the impacts of electrolyzer oversizing and time-varying electricity prices on levelized production cost. Ref.~\cite{he2026unlocking} developed an optimal sizing model for a RES-driven \ce{H2}-DRI-EAF system, showing that multi-stage production flexibility can lower the levelized cost of steel production while inducing differentiated overcapacity in the electrolyzer, SF, and EAF. 
However, in these studies, higher production flexibility is achieved through the oversizing of key production units. Although such oversizing may reduce levelized production costs, it may also lead to redundant investment and overcapacity, thereby intensifying market competition \cite{organisation2015excess,sun2025excess}. Therefore, it is necessary to decouple maximum hourly production-rate limits from annual production capacities and treat them as independent decision variables, thereby increasing production rate flexibility without expanding annual production capacities and improving investment efficiency. \par
Additionally, techno-economic analyses of the integrated \ce{H2}-DRI-EAF-MeOH system remain limited, and optimal sizing of this system has received even less attention. For such a multi-product integrated system, optimal sizing entails jointly balancing capital investment, operating costs, and multiple revenue streams, as single-product cost metrics are insufficient to evaluate overall investment efficiency \cite{yu2025novel}. However, existing studies on green hydrogen production, hydrogen-based steelmaking, and green methanol synthesis generally analyze techno-economic performance using net present value (NPV), ANR, and product-specific levelized-cost metrics, including the levelized costs of electricity, hydrogen, steel, and MeOH \cite{he2026unlocking,gu2022techno,delapedra2022methods}. For example, Ref.~\cite{devlin2023global} assessed the regional cost competitiveness of green steel production using an NPV-based lifetime-cost assessment, while Ref.~\cite{su2025capacity} optimized BS capacity by maximizing ANR. Refs.~\cite{darling2011assumptions,henry2023techno,roulier2026hybrid,chen2021renewable} further evaluated photovoltaic power generation, green hydrogen production, steel production, and MeOH production using the levelized costs of electricity, hydrogen, steel, and MeOH, respectively. However, NPV and ANR evaluate absolute returns without normalizing them by the associated capital investment and may therefore favor capital-intensive sizing solutions with lower investment efficiency. Meanwhile, levelized-cost metrics focus on the unit costs of individual products and are insufficient to represent the overall investment efficiency of an integrated \ce{H2}-DRI-EAF-MeOH system with multiple revenue streams. Therefore, it is necessary to formulate a fractional programming model for optimal sizing, in which ANROI is adopted as the fractional objective to evaluate system economic performance. \par
\subsection{Research gap and contribution}
Overall, techno-economic analyses of the \ce{H2}-DRI-EAF-MeOH system remain limited, particularly in jointly considering hourly production rate flexibility and investment efficiency in optimal sizing. To address these gaps, this study develops a fractional programming-based optimal sizing model for the \ce{H2}-DRI-EAF-MeOH system. The main contributions are as follows: \par
1) The maximum hourly production rates of the SF, EAF, and MSR are decoupled from annual production capacities and formulated as independent decision variables. Theoretical analysis demonstrates that increasing hourly production rate limits expands the multi-period operating feasible region, increases the theoretical upper bound on RES matching, and provides an investment-efficiency condition for increasing production rate limits. \par
2) A mixed-integer linear fractional programming (MILFP) model is established for optimal sizing, with ANROI as the fractional objective. The model is solved using the Dinkelbach method, and system performance is evaluated using internal rate of return (IRR), discounted payback period (DPP), renewable energy curtailment rate (RECR), and net on-site renewable energy utilization rate (OSREUR).\par
3) Case studies are conducted under on-grid, off-grid, and off-grid zero-carbon (no direct carbon emissions) scenarios to validate the proposed fractional programming model incorporating hourly production rate flexibility. The simulation results demonstrate that the proposed model outperforms the annualized net return maximization benchmark in terms of IRR, DPP, RECR, and OSREUR.
\subsection{Paper organization}
The rest of the paper is organized as follows: Section 2 describes the architecture and operation models of the \ce{H2}-DRI-EAF-MeOH system. Section 3 presents the MILFP model and evaluation metrics. The case study is carried out in Section 4. Conclusions are drawn in Section 5.
\section{The architecture and operation models of the \ce{H2}-DRI-EAF-MeOH system}
\label{sec2}
\subsection{The architecture of the \ce{H2}-DRI-EAF-MeOH system}
\label{subsec1}
The architecture of the proposed \ce{H2}-DRI-EAF-MeOH system is illustrated in Fig. 1. The overall production process is as follows: hydrogen is used to reduce \ce{Fe2O3} to obtain DRI; the DRI and scrap steel are subsequently melted in the EAF, where carbon powder is added to produce crude steel; meanwhile, the generated \ce{CO2} is captured and reacted with hydrogen to synthesize MeOH in the MSR, thereby avoiding direct \ce{CO2} emissions. The energy-supply side comprises RES and the external power grid. The load side includes alkaline electrolyzer (AE), high- and low-temperature resistance heaters, SF, EAF, and MSR. The storage units include BS, hydrogen storage tank (HS), low-temperature thermal storage tank (Lts), \ce{CO2} storage tank (CS), and DRI storage bin (DS). \par
\begin{figure*}[t!]
\centering
\includegraphics[width=1\textwidth]{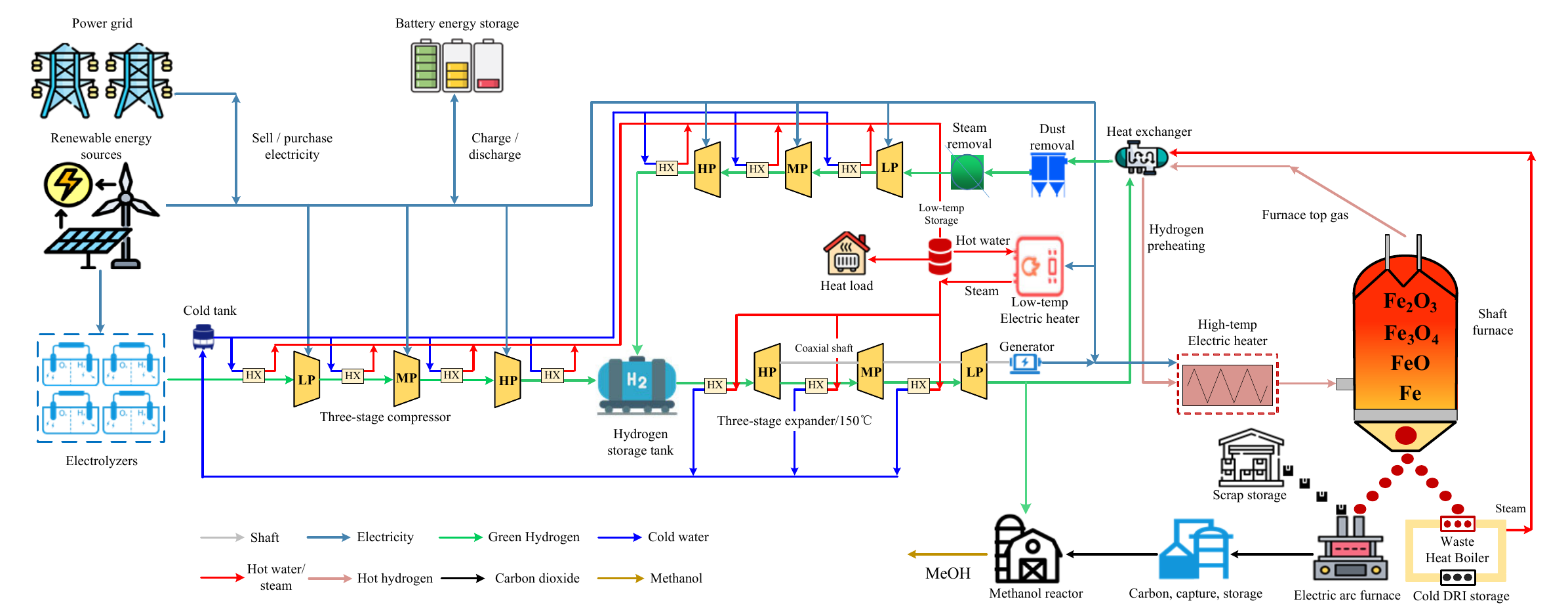}
\setlength{\abovecaptionskip}{-0.1cm}  
\setlength{\belowcaptionskip}{-0.1cm} 
\caption{The architecture of hydrogen-based direct reduced iron-electric arc furnace system integrated with methanol synthesis}
\captionsetup{justification=centering}
\vspace{-0.4cm}
\end{figure*}\par
\subsection{Operation models}
Over a one-year horizon discretized into hourly intervals, the time set is defined as ${T}=\{1,2,\ldots,8760\}$. The operation models of the proposed \ce{H2}-DRI-EAF-MeOH system are formulated as follows.
\\
1) Modeling of RES generation
\begin{subequations}
\begin{equation}
P_{\mathrm{RES}}^{t}
=P_{\mathrm{W}}^{t}+P_{\mathrm{S}}^{t}-P_{\mathrm{curt}}^{t},
\end{equation}
\begin{equation}
P_{\mathrm{W}}^{t}
=
\psi_{\mathrm{W}}^{t}
C_{\mathrm{W}},
\end{equation}
\begin{equation}
P_{\mathrm{S}}^{t}
=
\psi_{\mathrm{S}}^{t}
C_{\mathrm{S}}.
\end{equation}
\begin{equation}
0
\leq
P_{\mathrm{curt}}^{t}
\leq
P_{\mathrm{W}}^{t}
+
P_{\mathrm{S}}^{t},
\end{equation}
\end{subequations}
$P_{\mathrm{RES}}^{t}$ is the RES power output;
$P_{\mathrm{W}}^{t}$ and $P_{\mathrm{S}}^{t}$ represent the wind and PV power outputs, respectively;
$P_{\mathrm{curt}}^{t}$ is the curtailed power;
$\psi_{\mathrm{W}}^{t}$ and $\psi_{\mathrm{S}}^{t}$ are the wind and PV capacity factors, respectively;
$C_{\mathrm{W}}$ and $C_{\mathrm{S}}$ are the installed capacities of wind power and PV power, respectively.
\\
2) Modeling of hydrogen production and storage \par
The AE produces hydrogen through water electrolysis, and its operation is modeled as follows:
\begin{subequations}
\begin{equation}
M_{\mathrm{P2H}}^{t}
=
\psi_{\mathrm{P2H}}
P_{\mathrm{AE}}^{t},
\end{equation}
\begin{equation}
\psi_{\mathrm{AE,dn}}
C_{\mathrm{AE}}
\leq
P_{\mathrm{AE}}^{t}
\leq
\psi_{\mathrm{AE,up}}
C_{\mathrm{AE}}.
\end{equation}
\end{subequations}
$M_{\mathrm{P2H}}^{t}$ is the hydrogen production rate of the AE; $\psi_{\mathrm{P2H}}$ is the power-to-hydrogen conversion coefficient; 
$P_{\mathrm{AE}}^{t}$ is the power consumption of the AE; 
$C_{\mathrm{AE}}$ is the installed capacity of the AE; 
and $\psi_{\mathrm{AE,dn}}$ and $\psi_{\mathrm{AE,up}}$ are the lower and upper operating limits of the AE, respectively. \par
The HS balances hydrogen production and consumption, and its dynamics are formulated as follows:
\begin{subequations}
\begin{equation}
\mathrm{HS}^{t+1}
=
\mathrm{HS}^{t}
+
(M_{\mathrm{P2H}}^{t}
-
M_{\mathrm{HDRI}}^{t}
-
M_{\mathrm{MeOH,H_2}}^{t}) \Delta t,
\end{equation}
\begin{equation}
\psi_{\mathrm{HS,dn}}
C_{\mathrm{HS}}
\leq
\mathrm{HS}^{t}
\leq
\psi_{\mathrm{HS,up}}
C_{\mathrm{HS}},
\end{equation}
\begin{equation}
8
\left(
\psi_{\mathrm{THDRI}}
M_{\mathrm{SF},\max}
+
\psi_{\mathrm{H2M}}
M_{\mathrm{MSR},\max}
\right)
\leq
C_{\mathrm{HS}}.
\end{equation}
\begin{equation}
\left.
\mathrm{HS}^{t}
\right|_{t=1}
=
\left.
\mathrm{HS}^{t}
\right|_{t=|\mathrm{T}|+1}
=
\frac{1}{2}
C_{\mathrm{HS}}.
\end{equation}
\end{subequations}
$\mathrm{HS}^{t}$ is the hydrogen storage state; 
$M_{\mathrm{HDRI}}^{t}$ and $M_{\mathrm{MeOH,H_2}}^{t}$ are the hydrogen consumption rates for DRI production and MeOH synthesis, respectively;
$C_{\mathrm{HS}}$ is the HS capacity; 
$\psi_{\mathrm{HS,dn}}$ and $\psi_{\mathrm{HS,up}}$ are the lower and upper bounds of the HS, respectively; 
$\psi_{\mathrm{THDRI}}$ is the hydrogen demand coefficient for DRI production;  $\psi_{\mathrm{H2M}}$ is the hydrogen consumption coefficient of MeOH synthesis;
$M_{\mathrm{SF,max}}$ and $M_{\mathrm{MSR,max}}$ are the maximum hourly production rates of the SF and MSR, respectively.
\\
3) Modeling of BS \par
Unlike HS, the BS is subject to self-discharge and charging/discharging efficiency losses. Its operation is described as follows:
\begin{subequations}
\begin{equation}
\mathrm{BS}^{t+1}
=
\left(1-\psi_{\mathrm{BS}}\right)
\mathrm{BS}^{t}
+
(
\psi_{\mathrm{BS,ch}} P_{\mathrm{BS,ch}}^{t}
-
\frac{1}{\psi_{\mathrm{BS,dis}}}
P_{\mathrm{BS,dis}}^{t}
) \Delta t,
\end{equation}
\begin{equation}
\psi_{\mathrm{BS,dn}} C_{\mathrm{BS}}
\leq
\mathrm{BS}^{t}
\leq
\psi_{\mathrm{BS,up}} C_{\mathrm{BS}},
\end{equation}
\begin{equation}
\left.
\mathrm{BS}^{t}
\right|_{t=1}
=
\left.
\mathrm{BS}^{t}
\right|_{t=|\mathrm{T}|+1}
=
\frac{1}{2}
C_{\mathrm{BS}},
\end{equation}
\begin{equation}
0
\leq
P_{\mathrm{BS,ch}}^{t},
P_{\mathrm{BS,dis}}^{t}
\leq
\frac{C_{\mathrm{BS}}}{H_{\mathrm{BS}}}.
\end{equation}
\end{subequations}
$\mathrm{BS}^{t}$ is the stored energy; $\psi_{\mathrm{BS}}$ is the self-discharging rate; $\psi_{\mathrm{BS,ch}}$ and $\psi_{\mathrm{BS,dis}}$ are the charging and discharging efficiencies, respectively; 
$P_{\mathrm{BS,ch}}^{t}$ and $P_{\mathrm{BS,dis}}^{t}$ are the charging and discharging powers, respectively; 
$C_{\mathrm{BS}}$ is the capacity of the BS; 
$\psi_{\mathrm{BS,dn}}$ and $\psi_{\mathrm{BS,up}}$ are the lower and upper state-of-charge limits, respectively;
and $H_{\mathrm{BS}}$ is the rated discharge duration of the BS, which is set to 2 hours.
\\
4) Modeling of Lts \par
The Lts stores low-temperature heat from the AE, compressors, and the low-temperature electric heater and supplies heat for external demand and expander preheating. Its operation is modeled as follows:
\begin{subequations}
\begin{equation}
\mathrm{Lts}^{t+1}
=
\mathrm{Lts}^{t}
+
\left[
\left(
W_{\mathrm{AE}}^{t}
+
W_{\mathrm{comp}}^{t}
+
W_{\mathrm{Leh}}^{t}
\right)
\psi_{\mathrm{Lts,in}}
-
\frac{
W_{\mathrm{He}}^{t}
+
W_{\mathrm{Exp}}^{t}
}{
\psi_{\mathrm{Lts,out}}
}
\right]
\Delta t,
\end{equation}
\begin{equation}
\psi_{\mathrm{Lts,dn}}
C_{\mathrm{Lts}}
\leq
\mathrm{Lts}^{t}
\leq
\psi_{\mathrm{Lts,up}}
C_{\mathrm{Lts}},
\end{equation}
\begin{equation}
H_{\mathrm{Lts}}
W_{\mathrm{Exp},\max}
\leq
C_{\mathrm{Lts}},
\end{equation}
\begin{equation}
\left.
\mathrm{Lts}^{t}
\right|_{t=1}
=
\left.
\mathrm{Lts}^{t}
\right|_{t=|\mathrm{T}|+1}.
\end{equation}
\end{subequations}
$\mathrm{Lts}^{t}$ is the stored thermal energy; 
$C_{\mathrm{Lts}}$ is the thermal energy capacity of Lts; 
$\psi_{\mathrm{Lts,dn}}$ and $\psi_{\mathrm{Lts,up}}$ are the lower and upper state-of-charge limits of the Lts, respectively; 
$W_{\mathrm{AE}}^{t}$, $W_{\mathrm{comp}}^{t}$, and $W_{\mathrm{Leh}}^{t}$ are the heat inputs from AE, compressors, and the low-temperature electric heater, respectively; 
$\psi_{\mathrm{Lts,in}}$ and $\psi_{\mathrm{Lts,out}}$ are the charging and discharging efficiencies of the Lts, respectively; 
$W_{\mathrm{He}}^{t}$ and $W_{\mathrm{Exp}}^{t}$ are the external heat demand and expander preheating demand, respectively; 
$W_{\mathrm{Exp},\max}$ is the maximum expander preheating demand; 
$H_{Lts}$ is the minimum storage duration required for expander preheating.
\\
5) Modeling of SF \par
The SF reduces \ce{Fe2O3} to DRI using hydrogen at 1050~\textdegree C.
Before entering the SF, hydrogen is preheated using heat recovered from furnace top gas and waste heat boiler, and is then further heated by the high-temperature electric heater (HEH). The HEH is modeled as follows:
\begin{subequations}
\begin{equation}
P_{\mathrm{HEH}}^{t}
=
\psi_{\mathrm{Tth}}
M_{\mathrm{DRI}}^{t}
-
\frac{
\psi_{\mathrm{Th,re}}
\left(
W_{\mathrm{ftg}}^{t}
+
W_{\mathrm{whb}}^{t}
\right)
}{
\psi_{\mathrm{Eh}}
},
\end{equation}
\begin{equation}
W_{\mathrm{ftg}}^{t}
=
\psi_{\mathrm{ftg}}
M_{\mathrm{DRI}}^{t},
\end{equation}
\begin{equation}
W_{\mathrm{whb}}^{t}
=
\psi_{\mathrm{whb}}
M_{\mathrm{DS,in}}^{t}.
\end{equation}
\end{subequations}
$P_{\mathrm{HEH}}^{t}$ is the electric power consumption of the HEH; 
$M_{\mathrm{DRI}}^{t}$ is the DRI discharge rate of the SF; 
$W_{\mathrm{ftg}}^{t}$ and $W_{\mathrm{whb}}^{t}$ are the recovered heat from furnace top gas and from hot DRI by the waste heat boiler, respectively; 
$\psi_{\mathrm{Tth}}$ is the hydrogen-heating coefficient for DRI; 
$\psi_{\mathrm{ftg}}$ and $\psi_{\mathrm{whb}}$ are the heat recovery coefficients of furnace top gas and waste heat boiler, respectively; 
$\psi_{\mathrm{Th,re}}$ is the heat exchange coefficient; 
and $\psi_{\mathrm{Eh}}$ is the electric heating efficiency.\par
To capture the dynamic inertia of the hydrogen-based reduction of \ce{Fe2O3} in the SF, a first-order quasi-steady-state lag model \cite{ji2026demandresponsepotentialevaluation} is used to describe the hourly DRI discharge rate. 
\begin{subequations}
\begin{equation}
M_{\mathrm{DRI}}^{t+1}
=
M_{\mathrm{DQS}}^{t+1}
+
\left(
M_{\mathrm{DQS}}^{t}
-
M_{\mathrm{DQS}}^{t+1}
\right)
\alpha,
\quad
\alpha
=
e^{
-\frac{\Delta t}{T_{\mathrm{AE}}}
},
\end{equation}
\begin{equation}
\psi_{\mathrm{SF,min}}
M_{\mathrm{SF},\max}
-
\mathrm{M}
\left(
1
-
OP_{\mathrm{SF}}^{t}
\right)
\leq
M_{\mathrm{DQS}}^{t}
\leq
\psi_{\mathrm{SF,max}}
M_{\mathrm{SF},\max},
\end{equation}
\begin{equation}
0
\leq
M_{\mathrm{DRI}}^{t+1}
\leq
M_{\mathrm{SF},\max}
OP_{\mathrm{SF}}^{t},
\end{equation}
\begin{equation}
\psi_{\mathrm{SF,dr}}
M_{\mathrm{SF},\max}
\leq
M_{\mathrm{DQS}}^{t+1}
-
M_{\mathrm{DQS}}^{t}
\leq
\psi_{\mathrm{SF,ur}}
M_{\mathrm{SF},\max},
\end{equation}
\begin{equation}
M_{\mathrm{HDRI}}^{t}
=
\psi_{\mathrm{RHDRI}}
M_{\mathrm{DRI}}^{t},
\end{equation}
\begin{equation}
C_{\mathrm{SF}}
=
\sum_{t \in {T}}
M_{\mathrm{DRI}}^{t}\Delta t.
\end{equation}
\end{subequations}
$M_{\mathrm{DQS}}^{t}$ is the quasi-steady-state DRI discharge target of the SF; 
$\alpha$ is the lag factor; 
$T_{\mathrm{AE}}$ is the transition time constant; 
$\mathrm{M}$ is a sufficiently large positive number; 
$OP_{\mathrm{SF}}^{t}$ is a binary variable indicating the normal operating state of the SF;
$M_{\mathrm{SF},\max}$ is the maximum hourly production rate of the SF; 
$\psi_{\mathrm{SF,min}}$ and $\psi_{\mathrm{SF,max}}$ are the lower and upper production-rate coefficients of the SF, respectively; 
$\psi_{\mathrm{SF,dr}}$ and $\psi_{\mathrm{SF,ur}}$ are the down-ramp and up-ramp coefficients of the SF, respectively; 
$\psi_{\mathrm{RHDRI}}$ is the hydrogen consumption coefficient for DRI \cite{ji2025energy}; 
and $C_{\mathrm{SF}}$ is the annual DRI production capacity of the SF.\par
To describe the non-instantaneous start-up and shutdown processes of the SF, binary variables are introduced for start-up and shutdown commands, transition states, and normal operation. The corresponding constraints are formulated as follows:
\begin{subequations}
\begin{equation}
N_{\mathrm{SF,su}}
=
\left\lceil
\frac{
T_{\mathrm{SF,su}}
}{
\Delta t
}
\right\rceil,
\quad
N_{\mathrm{SF,sd}}
=
\left\lceil
\frac{
T_{\mathrm{SF,sd}}
}{
\Delta t
}
\right\rceil,
\end{equation}
\begin{equation}
SU_{\mathrm{SF}}^{t}
=
\sum_{\tau=\max\{1,t-N_{\mathrm{SF,su}}+1\}}^{t}
U_{\mathrm{SF}}^{\tau},
\end{equation}
\begin{equation}
SD_{\mathrm{SF}}^{t}
=
\sum_{\tau=\max\{1,t-N_{\mathrm{SF,sd}}+1\}}^{t}
D_{\mathrm{SF}}^{\tau},
\end{equation}
\begin{equation}
OP_{\mathrm{SF}}^{t}
=
OP_{\mathrm{SF}}^{t-1}
+
U_{\mathrm{SF}}^{t-N_{\mathrm{SF,su}}}
-
D_{\mathrm{SF}}^{t},
\end{equation}
\begin{equation}
OP_{\mathrm{SF}}^{t}
+
SU_{\mathrm{SF}}^{t}
+
SD_{\mathrm{SF}}^{t}
\leq
1,
\end{equation}
\begin{equation}
U_{\mathrm{SF}}^{t}
\leq
1
-
OP_{\mathrm{SF}}^{t-1}
-
SU_{\mathrm{SF}}^{t-1}
-
SD_{\mathrm{SF}}^{t-1},
\end{equation}
\begin{equation}
D_{\mathrm{SF}}^{t}
\leq
OP_{\mathrm{SF}}^{t-1},
\end{equation}
\begin{equation}
U_{\mathrm{SF}}^{t}
+
D_{\mathrm{SF}}^{t}
\leq
1.
\end{equation}
\end{subequations}
$U_{\mathrm{SF}}^{t}$ and $D_{\mathrm{SF}}^{t}$ are the start-up and shutdown commands of the SF, respectively; 
$SU_{\mathrm{SF}}^{t}$ and $SD_{\mathrm{SF}}^{t}$ are the start-up and shutdown transition states of the SF, respectively; 
$OP_{\mathrm{SF}}^{t}$ is the normal operating state of the SF; 
$T_{\mathrm{SF,su}}$ and $T_{\mathrm{SF,sd}}$ are the required start-up and shutdown transition times, respectively; 
and $N_{\mathrm{SF,su}}$ and $N_{\mathrm{SF,sd}}$ are the corresponding numbers of scheduling intervals.\par
The DS is used to mitigate temporal mismatch between DRI production in the SF and DRI consumption in the EAF, thereby decoupling upstream and downstream operations. Its inventory dynamics are formulated as follows:
\begin{subequations}
\begin{equation}
M_{\mathrm{DRI}}^{t}
=
M_{\mathrm{DRI,hot}}^{t}
+
M_{\mathrm{DS,in}}^{t},
\end{equation}
\begin{equation}
\mathrm{DS}^{t+1}
=
\mathrm{DS}^{t}
+
(M_{\mathrm{DS,in}}^{t}
-
M_{\mathrm{DRI,EAF}}^{t}
-
M_{\mathrm{DRI,s}}^{t})\Delta t,
\end{equation}
\begin{equation}
\psi_{\mathrm{DS,dn}}
C_{\mathrm{DS}}
\leq
\mathrm{DS}^{t}
\leq
\psi_{\mathrm{DS,up}}
C_{\mathrm{DS}},
\end{equation}
\begin{equation}
\left.
\mathrm{DS}^{t}
\right|_{t=1}
=
\left.
\mathrm{DS}^{t}
\right|_{t=|\mathrm{T}|+1}
=
\frac{1}{2}
C_{\mathrm{DS}}.
\end{equation}
\end{subequations}
$M_{\mathrm{DRI,hot}}^{t}$ is the hot DRI directly fed to the EAF; $M_{\mathrm{DS,in}}^{t}$ is the DRI sent to the DS; $\mathrm{DS}^{t}$ is the DRI storage state; $M_{\mathrm{DRI,EAF}}^{t}$ is the stored DRI supplied from the DS to the EAF; 
$M_{\mathrm{DRI,s}}^{t}$ is the stored DRI sold externally; $C_{\mathrm{DS}}$ is the DRI storage capacity; and $\psi_{\mathrm{DS,dn}}$ and $\psi_{\mathrm{DS,up}}$ are the lower and upper inventory limits of the DS, respectively.
\\
6) Modeling of EAF \par
The EAF melts DRI and scrap steel to crude steel. The EAF operating feasible region is adopted from Ref.~\cite{ji2026processawaredemandresponseevaluation} and represented in a compact form as follows:
\begin{subequations}
\begin{equation}
\boldsymbol{v}^{t}
=
\left[
x_{1}^{t},
x_{2}^{t},
x_{3}^{t},
P_{\mathrm{EAF}}^{t}
\right]^{\mathrm{T}},
\end{equation}
\begin{equation}
\boldsymbol{v}^{t}
\in
\Gamma_{\mathrm{EAF}},\quad \psi_{\mathrm{M,3}}x_{3}^{t}
\leq
0.2M_{\mathrm{steel}}^{t},
\end{equation}
\begin{equation}
C_{\mathrm{EAF}}
=
\sum_{t \in {T}}
M_{\mathrm{steel}}^{t} \Delta t.
\end{equation}
\end{subequations}
$\boldsymbol{v}^{t}$ is the EAF operating vector; 
$x_{i}^{t}$ is the sensible heat from the $i$-th material; 
$\psi_{\mathrm{M,3}}$ is the iron-mass coefficient for scrap steel, and 0.2 denotes the maximum allowable scrap-to-steel mass ratio; $P_{\mathrm{EAF}}^{t}$ is the EAF power; 
$\Gamma_{\mathrm{EAF}}$ denotes the EAF operating feasible region; 
$M_{\mathrm{steel}}^{t}$ is the crude steel production rate; 
and $C_{\mathrm{EAF}}$ is the annual crude steel production capacity.
\begin{subequations}
\begin{equation}
M_{\mathrm{EAF,CO_2}}^{t}
=
\psi_{\mathrm{C_c}} M_{\mathrm{Carbon}}^{t}
+
\psi_{\mathrm{C_l}} M_{\mathrm{limestone}}^{t}
+
\psi_{\mathrm{C_s}} M_{\mathrm{Scrap}}^{t},
\end{equation}
\begin{equation}
M_{\mathrm{Carbon}}^{t}
=
\psi_{\mathrm{C,DRI}} M_{\mathrm{steel}}^{t},
\end{equation}
\begin{equation}
M_{\mathrm{limestone}}^{t}
=
\psi_{\mathrm{limestone}} M_{\mathrm{steel}}^{t},
\end{equation}
\begin{equation}
M_{\mathrm{EAF,CO_2}}^{t}
=
M_{\mathrm{CCU}}^{t}
+
M_{\mathrm{emit}}^{t}.
\end{equation}
\end{subequations}
$M_{\mathrm{EAF,CO_2}}^{t}$ is the $\mathrm{CO_2}$ generation rate of the EAF;
$M_{\mathrm{Carbon}}^{t}$ ,$M_{\mathrm{limestone}}^{t}$, and $M_{\mathrm{Scrap}}^{t}$ are the consumption rates of carbon powder, limestone, and scrap steel, respectively; $\psi_{\mathrm{C_c}}$, $\psi_{\mathrm{C_l}}$, and $\psi_{\mathrm{C_s}}$ are the corresponding $\mathrm{CO_2}$ emission factors; $\psi_{\mathrm{limestone}}$ is the limestone consumption coefficient;
$M_{\mathrm{CCU}}^{t}$ is the captured $\mathrm{CO_2}$ rate; 
$M_{\mathrm{emit}}^{t}$ is the directly emitted $\mathrm{CO_2}$ rate.
\begin{subequations}
\begin{equation}
\mathrm{CS}^{t+1}
=
\mathrm{CS}^{t}
+
(M_{\mathrm{CCU}}^{t}
-
M_{\mathrm{MeOH,CO_2}}^{t})\Delta t,
\end{equation}
\begin{equation}
P_{\mathrm{CCR}}^{t}
=
\psi_{\mathrm{CCR}} M_{\mathrm{CCU}}^{t},
\end{equation}
\begin{equation}
M_{\mathrm{CCU}}^{t}
\leq
\mathrm{M} b_{\mathrm{CCU,MeOH}},
\end{equation}
\begin{equation}
\psi_{\mathrm{CS,dn}} C_{\mathrm{CS}}
\leq
\mathrm{CS}^{t}
\leq
\psi_{\mathrm{CS,up}} C_{\mathrm{CS}},
\end{equation}
\begin{equation}
8\psi_{\mathrm{C2M}}M_{\mathrm{MSR,max}}
\leq
C_{\mathrm{CS}},
\end{equation}
\begin{equation}
C_{\mathrm{CS}}
\leq
\mathrm{M} b_{\mathrm{CCU,MeOH}},
\end{equation}
\begin{equation}
\left.\mathrm{CS}^{t}\right|_{t=1}
=
\left.\mathrm{CS}^{t}\right|_{t=T+1}
=
\frac{1}{2}C_{\mathrm{CS}}.
\end{equation}
\end{subequations}
$\mathrm{CS}^{t}$ is the stored $\mathrm{CO_2}$;
$M_{\mathrm{CCU}}^{t}$ is the captured $\mathrm{CO_2}$ rate of the carbon capture and utilization (CCU); 
$M_{\mathrm{MeOH,CO_2}}^{t}$ is the $\mathrm{CO_2}$ consumption rate for MeOH synthesis; 
$P_{\mathrm{CCR}}^{t}$ is the electrical power consumption of CCU; 
$\psi_{\mathrm{CCR}}$ is the power consumption coefficient of CCU; 
$C_{\mathrm{CS}}$ is the $\mathrm{CO_2}$ storage capacity; 
$\psi_{\mathrm{CS,dn}}$ and $\psi_{\mathrm{CS,up}}$ are the lower and upper inventory limits of the CS, respectively; 
$b_{\mathrm{CCU,MeOH}}$ is a binary variable indicating the installation of the CCU and MSR systems; the CS capacity is required to cover 8 h of \ce{CO2} demand at the maximum MSR production rate; 
$\psi_{\mathrm{C2M}}$ is the $\mathrm{CO_2}$-to-MeOH demand coefficient; and $M_{\mathrm{MSR,max}}$ is the maximum production rate of the MSR. 
\\
7) Modeling of MSR \par
Given the operational similarities between methanol and ammonia synthesis \cite{xu2026flexibilizing} processes, a simplified scheduling model is formulated for the flexible operation of the MSR as follows:
\begin{subequations}
\begin{equation}
M_{\mathrm{MeOH,H_2}}^{t}
=
\psi_{\mathrm{H2M}} M_{\mathrm{MeOH}}^{t},
\end{equation}
\begin{equation}
M_{\mathrm{MeOH,CO_2}}^{t}
=
\psi_{\mathrm{C2M}} M_{\mathrm{MeOH}}^{t},
\end{equation}
\begin{equation}
\psi_{\mathrm{MSR,dn}} M_{\mathrm{MSR,max}}
\leq
M_{\mathrm{MeOH}}^{t}
\leq
\psi_{\mathrm{MSR,up}} M_{\mathrm{MSR,max}},
\end{equation}
\begin{equation}
\psi_{\mathrm{MSR,dr}} M_{\mathrm{MSR,max}}
\leq
M_{\mathrm{MeOH}}^{t+1}-M_{\mathrm{MeOH}}^{t}
\leq
\psi_{\mathrm{MSR,ur}} M_{\mathrm{MSR,max}},
\end{equation}
\begin{equation}
C_{\mathrm{MSR}}
=
\sum_{t\in{T}} M_{\mathrm{MeOH}}^{t} \Delta t,
\end{equation}
\begin{equation}
M_{\mathrm{MSR,max}}
\leq
\mathrm{M} b_{\mathrm{CCU,MeOH}},
\end{equation}
\begin{equation}
C_{\mathrm{CS}}
+
C_{\mathrm{MSR}}
+
\mathrm{M}\left(1-b_{\mathrm{CCU,MeOH}}\right)
\geq
\psi_{\mathrm{eps}}.
\end{equation}
\end{subequations}
$M_{\mathrm{MeOH}}^{t}$ is the MeOH production rate of the MSR; 
$\psi_{\mathrm{MSR,dn}}$ and $\psi_{\mathrm{MSR,up}}$ are the lower and upper production limits of the MSR, respectively;
$\psi_{\mathrm{MSR,dr}}$ and $\psi_{\mathrm{MSR,ur}}$ are the down-ramp and up-ramp rate coefficients of the MSR, respectively; 
$C_{\mathrm{MSR}}$ is the annual MeOH production capacity;
and $\psi_{\mathrm{eps}}$ is a small positive constant.
\\
8) Modeling of system power balance \par
The hourly power balance of the proposed \ce{H2}-DRI-EAF-MeOH system is formulated as follows:
\begin{subequations}
\begin{equation}
\begin{split}
P_{\mathrm{W}}^{t}
+
&
P_{\mathrm{S}}^{t}
+
P_{\mathrm{Buy}}^{t}
+
P_{\mathrm{exp}}^{t}
+
P_{\mathrm{BS,dis}}^{t}=P_{\mathrm{Sell}}^{t}
+
P_{\mathrm{AE}}^{t}
\\
&\quad
+
P_{\mathrm{BS,ch}}^{t}
+
P_{\mathrm{Leh}}^{t}
+
P_{\mathrm{comp}}^{t}
+
P_{\mathrm{HEH}}^{t}
+
P_{\mathrm{EAF}}^{t}
+
P_{\mathrm{CCR}}^{t}
+
P_{\mathrm{curt}}^{t},
\end{split}
\end{equation}
\begin{equation}
\sum_{t\in{T}}P_{\mathrm{Sell}}^{t}
\leq
\psi_{\mathrm{Sell}}
\sum_{t\in{T}}
\left(
P_{\mathrm{W}}^{t}
+
P_{\mathrm{S}}^{t}
\right),
\end{equation}
\begin{equation}
\sum_{t\in{T}}P_{\mathrm{Buy}}^{t}
\leq
\psi_{\mathrm{Grid,Buy,ratio}}
\sum_{t\in{T}}P_{\mathrm{load}}^{t}
\end{equation}
\begin{equation}
P_{\mathrm{load}}^{t}
=
P_{\mathrm{AE}}^{t}
+
P_{\mathrm{BS,ch}}^{t}
+
P_{\mathrm{Leh}}^{t}
+
P_{\mathrm{comp}}^{t}
+
P_{\mathrm{HEH}}^{t}
+
P_{\mathrm{EAF}}^{t}
+
P_{\mathrm{CCR}}^{t}
\end{equation}
\begin{equation}
0
\leq
P_{\mathrm{Sell}}^{t}
\leq
b_{\mathrm{grid}}^{t}\mathrm{M}
\end{equation}
\begin{equation}
0
\leq
P_{\mathrm{Buy}}^{t}
\leq
\left(1-b_{\mathrm{grid}}^{t}\right)\mathrm{M}
\end{equation}
\begin{equation}
P_{\mathrm{exp}}^{t}
=
\psi_{\mathrm{Eexp}}
\left(
M_{\mathrm{DRI}}^{t}\psi_{\mathrm{THDRI}}
+
M_{\mathrm{MeOH,H_2}}^{t}
\right)
\end{equation}
\begin{equation}
P_{\mathrm{comp}}^{t}
=
\psi_{\mathrm{ECcomp}}M_{\mathrm{DRI}}^{t}
+
\psi_{\mathrm{ERcomp}}M_{\mathrm{P2H}}^{t}
\end{equation}
\end{subequations}
$P_{\mathrm{Buy}}^{t}$ and $P_{\mathrm{Sell}}^{t}$ are the powers purchased from and sold to the external grid, respectively;
$P_{\mathrm{exp}}^{t}$ is the power output of the expander; 
$P_{\mathrm{Leh}}^{t}$ is the power consumption of the low-temperature electric heater; 
$P_{\mathrm{comp}}^{t}$ is the power consumption of the compressors;
$\psi_{\mathrm{Sell}}$ and $\psi_{\mathrm{Grid,Buy,ratio}}$ are the grid-export and grid import limit ratios, respectively; 
$b_{\mathrm{grid}}^{t}$ is the binary variable to prevent simultaneous electricity purchase and sale; 
$P_{\mathrm{load}}^{t}$ is the total electrical load of the system; 
$\psi_{\mathrm{Eexp}}$ is the power generation coefficient of the expansion;
$\psi_{\mathrm{ECcomp}}$ and $\psi_{\mathrm{ERcomp}}$ are the power consumption coefficients of circulating compressor and reducing compressor, respectively. 
\section{The planning model for the integrated \ce{H2}-DRI-EAF-MeOH system}
\subsection{Objective function and economic model}
To evaluate the overall investment return of the integrated \ce{H2}-DRI-EAF-MeOH system, the ANROI is adopted as the optimization objective:
\begin{equation}
\max_{\boldsymbol{z}}\ \mathrm{ANROI}
=
\frac{\mathrm{ANR}}{\mathrm{TCI}}
\end{equation}\par
The ANR is defined as the difference between total annual revenue and total annualized cost:
\begin{equation}
\begin{split}
\mathrm{ANR}
={}&
AR_{\mathrm{Steel}}
+
AR_{\mathrm{DRI}}
+
AR_{\mathrm{MeOH}}
+
AR_{\mathrm{Elec}}
+
AR_{\mathrm{He}}
+
AR_{\mathrm{Ce}}
\\
&-
AC_{\mathrm{Inv}}
-
AC_{\mathrm{OM}}
-
AC_{\mathrm{Deg}}
-
AC_{\mathrm{Ore}}
-
AC_{\mathrm{Scrap}}
-
AC_{\mathrm{Wat}} .
\end{split}
\end{equation}\par
The ANR revenue streams comprise product sales, net electricity trading, heat supply, and carbon allowance:
\begin{subequations}
\begin{equation}
AR_{\mathrm{Steel}}
=
\psi_{\mathrm{Steel,pri}}
\sum_{t\in{T}}
M_{\mathrm{Steel}}^{t}
\end{equation}
\begin{equation}
AR_{\mathrm{DRI}}
=
\psi_{\mathrm{DRI,pri}}
\sum_{t\in{T}}
M_{\mathrm{DRI,s}}^{t}
\end{equation}
\begin{equation}
AR_{\mathrm{MeOH}}
=
\psi_{\mathrm{MeOH,pri}}
\sum_{t\in{T}}
M_{\mathrm{MeOH}}^{t}
\end{equation}
\begin{equation}
AR_{\mathrm{Elec}}
={}
\sum_{t\in{T}}
(\psi_{\mathrm{Sell,pri}}^{t}
P_{\mathrm{Sell}}^{t}
-
\psi_{\mathrm{Buy,pri}}^{t}
P_{\mathrm{Buy}}^{t})-
\sum_{\tau \in\mathcal{M}}
\psi_{\mathrm{Cap}}
P_{\mathrm{dem}}^{\tau}
\end{equation}
\begin{equation}
P_{\mathrm{dem}}^{\tau}
\geq
P_{\mathrm{Buy}}^{t},
\qquad
\forall t\in{T}^{\tau},
\quad
\forall \tau\in{T}^{\mathrm{Mon}}
\end{equation}
\begin{equation}
AR_{\mathrm{He}}
=
\psi_{\mathrm{He,pri}}
\sum_{t\in {T}}
W_{\mathrm{He}}^{t}
\end{equation}
\begin{equation}
AR_{\mathrm{Ce}}
=
\psi_{\mathrm{Cb,pri}}
\sum_{t\in{T}}
\left(
\psi_{\mathrm{Cq,DRI}}
M_{\mathrm{DRI}}^{t}
+
\psi_{\mathrm{Cq,Steel}}
M_{\mathrm{Steel}}^{t}
-
M_{\mathrm{emit}}^{t}
\right)
\end{equation}
\end{subequations}
$\boldsymbol{z}$ denotes the decision variables; 
$\mathrm{ANROI}$ is the annualized net return on investment;
$\mathrm{ANR}$ is the annualized net return;
$\mathrm{TCI}$ is the total capital investment;
$AR_{\mathrm{Steel}}$, $AR_{\mathrm{DRI}}$, and $AR_{\mathrm{MeOH}}$ are the annual revenues from crude steel, DRI, and MeOH, respectively;
$AR_{\mathrm{Elec}}$ is the annual net revenue from grid electricity trading; 
$AR_{\mathrm{He}}$ is the annual heat supply revenue;
$AR_{\mathrm{Ce}}$ is the annual carbon allowance revenue;
$AC_{\mathrm{Inv}}$ is the annualized capital cost;
$AC_{\mathrm{OM}}$ is the operation and maintenance cost;
$AC_{\mathrm{Deg}}$ is the battery degradation cost;
and $AC_{\mathrm{Ore}}$, $AC_{\mathrm{Scrap}}$, and $AC_{\mathrm{Wat}}$ are the annual costs of iron ore, scrap steel, and industrial water, respectively.
$\psi_{\mathrm{Steel,pri}}$, $\psi_{\mathrm{DRI,pri}}$, and $\psi_{\mathrm{MeOH,pri}}$ are the prices of crude steel, DRI, and MeOH, respectively;
$\psi_{\mathrm{Sell,pri}}^{t}$ and $\psi_{\mathrm{Buy,pri}}^{t}$ are the electricity selling and purchasing prices, respectively; 
$\psi_{\mathrm{Cap}}$ is the capacity price;
$P_{\mathrm{dem}}^{\tau}$ is the monthly peak electricity demand in month $\tau$;
$\mathrm{Mon}$ is the set of months;
$\psi_{\mathrm{He,pri}}$ is the heat selling price;
$\psi_{\mathrm{Cb,pri}}$ is the carbon allowance price;
and $\psi_{\mathrm{Cq,DRI}}$ and $\psi_{\mathrm{Cq,Steel}}$ are the carbon allowances allocated to the ironmaking and steelmaking processes, respectively; the two allowances are mutually exclusive and do not involve double counting. \par
The annualized capital cost is calculated by applying the capital recovery factor to the total capital investment:
\begin{subequations}
\begin{equation}
AC_{\mathrm{Inv}}
=
\mathrm{CRF}(r,Y)
\mathrm{TCI}
\end{equation}
\begin{equation}
\mathrm{CRF}(r,Y)
=
\frac{
r(1+r)^{Y}
}{
(1+r)^{Y}-1
}
\end{equation}
\begin{equation}
\mathrm{TCI}
=
\sum_{j\in\Omega_{\mathrm{Inv}}}
I_{\mathrm{Inv}}^{j}C^{j}
+
\mathrm{TCI}_{\mathrm{DCP}}
+
\mathrm{TCI}_{\mathrm{ATS}}
\end{equation}
\begin{equation}
\Omega_{\mathrm{Inv}}
=
\left\{
\mathrm{W,S,AE,BS,HS,Lts,CS,DS}
\right\}
\end{equation}
\end{subequations} \par
For the SF, EAF, and MSR, the capital investments associated with annual production capacity and maximum hourly production rate are as follows:
\begin{subequations}
\begin{equation}
\mathrm{TCI}_{\mathrm{DCP}}
=
\sum_{k\in\Omega_{\mathrm{DCP}}}
(
c_{\mathrm{Inv}}^{k}
(
M_{k,\max}
-
\frac{C^{k}}{8000}
)
+
I_{\mathrm{Inv}}^{k}C^{k})
\end{equation}
\begin{equation}
\Omega_{\mathrm{DCP}}
=
\left\{
\mathrm{SF,EAF,MSR}
\right\}
\end{equation}
\begin{equation}
\frac{C^{k}}{8000}
\leq
M_{k,\max}
\leq
\frac{C^{k}}{4000},
\forall k\in\Omega_{\mathrm{DCP}}
\end{equation}
\end{subequations}\par
The investment costs for the remaining auxiliary equipment are calculated as follows:
\begin{subequations}
\begin{equation}
\mathrm{TCI}_{\mathrm{ATS}}
=
\sum_{l\in\Omega_{\mathrm{ATS}}}
I_{l}
\end{equation}
\begin{equation}
\Omega_{\mathrm{ATS}}
=
\left\{
\mathrm{Comp,Exp,HEH,Leh}
\right\}
\end{equation}
\end{subequations}\par
Additionally, the costs of operation and maintenance, battery degradation, and raw materials are calculated as follows:
\begin{subequations}
\begin{equation}
AC_{\mathrm{OM}}
=
\sum_{j\in\Omega_{\mathrm{Inv}}}
\psi_{\mathrm{OM}}^{j}
I_{\mathrm{Inv}}^{j}C^{j}
+ \sum_{k\in\Omega_{\mathrm{DCP}}}
\psi_{\mathrm{OM}}^{k}
I_{\mathrm{Inv}}^{k}C^{k}
\end{equation}
\begin{equation}
AC_{\mathrm{Deg}}
=
\psi_{\mathrm{Deg}}
\sum_{t\in{T}}
P_{\mathrm{BS,dis}}^{t}
\end{equation}
\begin{equation}
AC_{\mathrm{Ore}}
=
\psi_{\mathrm{OD}}
\sum_{t\in{T}}
M_{\mathrm{DRI}}^{t}
\end{equation}
\begin{equation}
AC_{\mathrm{Scrap}}
=
\psi_{\mathrm{Sp,prc}}
\sum_{t\in{T}}
M_{\mathrm{Scrap}}^{t}
\end{equation}
\begin{equation}
AC_{\mathrm{Wa}}
=
\psi_{\mathrm{Wa,pri}}
\gamma_{\mathrm{W2H}}
\sum_{t\in{T}}
M_{\mathrm{P2H}}^{t}
\end{equation}
\end{subequations}
$\mathrm{CRF}(r,Y)$ is the capital recovery factor;
$r$ is the discount rate;
$Y$ is the project lifetime;
$\Omega_{\mathrm{Inv}}$ is the set of units with capacity-based investment costs;
$I_{\mathrm{Inv}}^{j}$ is the unit investment cost of unit $j$;
$C^{j}$ is the installed capacity of unit $j$;
$\mathrm{TCI}_{\mathrm{DCP}}$ is the investment cost of production units with decoupled annual capacity and hourly production-rate limits;
$\Omega_{\mathrm{DCP}}$ is the set of these production units; 
$I_{\mathrm{Inv}}^{k}$ is the unit investment cost of annual production capacity for unit $k$;
$c_{\mathrm{Inv}}^{k}$ is the unit investment cost associated with the maximum hourly production rate of unit $k$;
$M_{k,\max}$ is the maximum hourly production rate;
$\mathrm{TCI}_{\mathrm{ATS}}$ is the total investment in auxiliary units;
and $\Omega_{\mathrm{ATS}}$ is the set of auxiliary units.
$\psi_{\mathrm{OM}}$ is the operation and maintenance cost coefficient;
$\psi_{\mathrm{Deg}}$ is the battery degradation cost coefficient; 
$\psi_{\mathrm{OD}}$ is the iron ore cost per tonne of DRI;
$\psi_{\mathrm{Scrap}}$ is the scrap steel price;
$\psi_{\mathrm{Wa,pri}}$ is the industrial water price; and $\gamma_{\mathrm{W2H}}$ is the water consumption per tonne of hydrogen.
The ANROI is defined as a fractional objective that jointly incorporates multiple revenue streams and total capital investment. Although both $\mathrm{ANR}$ and $\mathrm{TCI}$ are linear functions of the decision variables, their ratio makes the optimization problem nonlinear and nonconvex, so it cannot be directly handled by standard MILP solvers. Accordingly, the proposed model is formulated as a MILFP problem and solved through a sequence of parameterized MILP subproblems.
\subsection{MILFP formulation and Dinkelbach solution method}
Based on the operational constraints and economic model, the optimal sizing and hourly operation decisions are integrated into a MILFP model. Its compact formulation is given as follows:
\begin{subequations}
\begin{equation}
\max_{\boldsymbol{z}}
\quad
\mathrm{ANROI}(\boldsymbol{z})
=
\frac{
\mathrm{ANR}(\boldsymbol{z})
}{
\mathrm{TCI}(\boldsymbol{z})
},
\end{equation}
\begin{equation}
\mathrm{s.t.}
\quad
\boldsymbol{z}
\in
\mathcal{R},
\end{equation}
\begin{equation}
\mathcal{R}
=
\left\{
\boldsymbol{z}
\mid
\text{constraints (1)--(21)}
\right\}.
\end{equation}
\end{subequations}
$\boldsymbol{z}$ denotes the complete decision-variable vector, and $\mathcal{R}$ denotes the feasible region defined by constraints (1)--(21).\par
To solve the ANROI-based MILFP model, the Dinkelbach method \cite{you2009dinkelbach,yu2025novel} is adopted to transform the fractional objective into a sequence of parameterized MILP subproblems. For a given parameter $\lambda$, the auxiliary function is defined as
\begin{equation}
\phi(\lambda)
=
\max_{\boldsymbol{z}\in\mathcal{R}}
\left\{
\mathrm{ANR}(\boldsymbol{z})
-
\lambda
\mathrm{TCI}(\boldsymbol{z})
\right\}.
\end{equation} \par
For a fixed $\lambda$, the objective function is linear and the feasible region remains unchanged; hence, each subproblem can be solved as a standard MILP problem. The optimal ANROI $\lambda^{*}$ is obtained when the auxiliary function satisfies
\begin{equation}
\phi(\lambda^{*})=0.
\end{equation}
 If $\phi(\lambda)>0$, the current value of $\lambda$ is lower than the optimal ANROI; conversely, if $\phi(\lambda)<0$, it is higher than the optimal ANROI.
The Dinkelbach-based solution procedure is summarized in pseudocode in Algorithm~\ref{alg:DKB}.
\begin{algorithm}[t]
\caption{Dinkelbach-based algorithm for the MILFP model}
\label{alg:DKB}
\begin{algorithmic}[1]
\Require $\lambda_1=0$, $\varepsilon=10^{-4}$, and $I_{\max}=50$
\Ensure Optimal solution $\boldsymbol{z}^{*}$ and ANROI $\lambda^{*}$
\For{$n=1,\ldots,I_{\max}$}
    \State Solve MILP subproblem for $\lambda_n$
    \[
    \boldsymbol{z}_n
    \in
    \arg\max_{\boldsymbol{z}\in\mathcal{R}}
    \left\{
    \mathrm{ANR}(\boldsymbol{z})
    -
    \lambda_n\mathrm{TCI}(\boldsymbol{z})
    \right\}
    \]
    \State Compute the residual
    \[
    r_n=
    \mathrm{ANR}(\boldsymbol{z}_n)
    -
    \lambda_n\mathrm{TCI}(\boldsymbol{z}_n)
    \]
    \If{$|r_n|\leq\varepsilon$}
        \State \Return $\boldsymbol{z}^{*}\gets\boldsymbol{z}_n,\ 
        \lambda^{*}\gets\lambda_n$
    \EndIf
    \State Update
    \[
    \lambda_{n+1}
    \gets
    \frac{\mathrm{ANR}(\boldsymbol{z}_n)}
    {\mathrm{TCI}(\boldsymbol{z}_n)}
    \]
\EndFor
\State \Return $\boldsymbol{z}^{*}\gets\boldsymbol{z}_n,\ 
\lambda^{*}\gets\lambda_{n+1}$
\end{algorithmic}
\end{algorithm}
\subsection{Mathematical characterization of hourly production rate flexibility}
For the SF, EAF, and MSR, annual production capacities are decoupled from maximum hourly production rates and treated as separate sizing variables, allowing annual production scale and hourly production capability to be independently determined. To characterize the mathematical role of hourly production-rate flexibility, the following lemmas are introduced: 
\begin{lemma}\textbf{Expansion of the temporal operating region}
\par 
Let $\Omega_{\mathrm{AP}}$ denote the set of adjustable production units. $\mathcal{F}\left(\boldsymbol{C},
\boldsymbol{M}_{\max}\right)$ denotes the multi-period feasible region for fixed annual production capacities and fixed capacities of all other system components. 
Suppose that increasing $M_{k,\max}$ only relaxes the upper bound on the hourly production of unit $k$, while the annual production capacities and all other constraint parameters remain unchanged. \\
If
\[
M_{k,\max}^{(1)}
\leq
M_{k,\max}^{(2)},
\qquad
\forall k\in\Omega_{\mathrm{AP}},
\]
then
\[
\mathcal{F}
\left(
\boldsymbol{C},
\boldsymbol{M}_{\max}^{(1)}
\right)
\subseteq
\mathcal{F}
\left(
\boldsymbol{C},
\boldsymbol{M}_{\max}^{(2)}
\right).
\]
\end{lemma}
\par\noindent\textbf{Proof.}\par
For any schedule feasible under
$\boldsymbol{M}_{\max}^{(1)}$, the hourly production variables satisfy
\[
M_k^t
\leq
M_{k,\max}^{(1)}
\leq
M_{k,\max}^{(2)},
\qquad
\forall k\in\Omega_{\mathrm{AP}},
\quad
\forall t\in{T}.
\]
Therefore, the same schedule also satisfies the relaxed hourly production limits associated with
$\boldsymbol{M}_{\max}^{(2)}$. Since the annual production constraints and all other operating constraints remain unchanged, every point in
$\mathcal{F}(\boldsymbol{C},\boldsymbol{M}_{\max}^{(1)})$
also belongs to
$\mathcal{F}(\boldsymbol{C},\boldsymbol{M}_{\max}^{(2)})$.
Thus, a higher hourly production limit expands temporal flexibility without changing annual capacity.
\begin{lemma}\textbf{RES matching properties of hourly production-rate flexibility}
\par Assume that the annual production capacities of the SF, EAF, and MSR are fixed, while their maximum hourly production rates are treated as adjustable variables. The available on-site renewable power is defined as
\begin{equation}
P_{\mathrm{R}}^{t}
=
P_{\mathrm{W}}^{t}
+
P_{\mathrm{S}}^{t}.
\end{equation}
For each unit $q\in\Omega_{\mathrm{AP}}$, $M_{q,\max}$ is the maximum hourly production rate, $C_q$ is the annual production capacity, and $\ell_q$ is the equivalent electricity consumption per unit of production. The aggregate upper bound on the hourly production load is defined as
\begin{equation}
x
=
P_{\mathrm{Prod,max}}
=
\sum_{q\in\Omega_{\mathrm{AP}}}
\ell_q M_{q,\max}.
\end{equation}
The annual equivalent production electricity demand is
\begin{equation}
E_{\mathrm{Prod}}
=
\sum_{q\in\Omega_{\mathrm{AP}}}
\ell_q C_q.
\end{equation}\par
Under a simplified aggregate representation without storage-based energy shifting, the theoretical upper bound on the amount of on-site RES electricity directly matched by production loads is jointly constrained by the annual production electricity demand and the aggregate hourly production-load bound, and is expressed as follows: 
\begin{equation}
E_{\mathrm{match}}(x)
=
\min
\left\{
E_{\mathrm{Prod}},
\,
\sum_{t\in\mathrm{T}}
\min
\left(
P_{\mathrm{R}}^{t},
x
\right)\Delta t
\right\}.
\end{equation}
Given the above definitions, $E_{\mathrm{match}}(x)$ is nondecreasing, piecewise linear, and concave with respect to the hourly production load boundary $x$. Before the annual production electricity demand is fully matched, its right-hand derivative is
\begin{equation}
\frac{\mathrm{d}^{+}E_{\mathrm{match}}(x)}
{\mathrm{d}x}
=
N
\left(
P_{\mathrm{R}}^{t}>x
\right) \Delta t,
\end{equation}
$N(P_{\mathrm{R}}^{t}>x)$ denotes the number of periods in which the available on-site renewable power exceeds the hourly production-load boundary $x$.
\par\noindent\textbf{Proof.}\\
Let
\begin{equation}
G(x)
=
\sum_{t\in\mathrm{T}}
\min
\left(
P_{\mathrm{R}}^{t},x
\right) \Delta t.
\end{equation}
For each period $t$, $P_{\mathrm{R}}^{t}$ is fixed with respect to $x$. Therefore, the function $\min(P_{\mathrm{R}}^{t},x)$ can be written as
\begin{equation}
\min(P_{\mathrm{R}}^{t},x)
=
\begin{cases}
x, & x<P_{\mathrm{R}}^{t},\\
P_{\mathrm{R}}^{t}, & x\geq P_{\mathrm{R}}^{t}.
\end{cases}
\end{equation}
Therefore, it is nondecreasing and piecewise linear in $x$. Moreover, its right-hand slope decreases from 1 to 0 when $x$ reaches $P_{\mathrm{R}}^{t}$, which implies concavity.
\begin{equation}
E_{\mathrm{match}}(x)
=
\min
\left\{
E_{\mathrm{Prod}},
G(x)
\right\},
\end{equation}
Since $G(x)$ is nondecreasing, piecewise linear, and concave, truncating it at the constant upper bound $E_{\mathrm{Prod}}$ preserves these properties. Specifically, $E_{\mathrm{match}}(x)$ follows $G(x)$ before the annual production electricity demand is fully matched and remains constant thereafter. Therefore, $E_{\mathrm{match}}(x)$ is also nondecreasing, piecewise linear, and concave in $x$. When $G(x)<E_{\mathrm{Prod}}$, the annual production electricity demand has not been fully matched, and thus $E_{\mathrm{match}}(x)=G(x)$. Increasing $x$ improves RES matching only in periods where the available on-site renewable power exceeds the current production-load boundary, i.e.,
\begin{equation}
P_{\mathrm{R}}^{t}>x.
\end{equation}
Therefore, the right-hand derivative of $E_{\mathrm{match}}(x)$ equals the number of such periods:
\begin{equation}
\frac{\mathrm{d}^{+}E_{\mathrm{match}}(x)}
{\mathrm{d}x}
=
N
\left(
P_{\mathrm{R}}^{t}>x
\right) \Delta t.
\end{equation}
As $x$ increases, fewer periods satisfy $P_{\mathrm{R}}^{t}>x$, so the marginal matching benefit decreases. Once $G(x)\geq E_{\mathrm{Prod}}$, the annual production electricity demand has been fully matched, and further increasing $x$ cannot increase $E_{\mathrm{match}}(x)$. This completes the proof.
\end{lemma}
\begin{lemma}\textbf{Investment-efficiency condition for increasing hourly production rate}
\par Let $\boldsymbol{z}^{*}$ and $\lambda^{*}$ denote the optimal solution and the optimal ANROI obtained from the Dinkelbach-based algorithm, respectively. Consider a feasible solution $\boldsymbol{z}'$ generated by a marginal increase in the maximum hourly production rate of production unit $k$. The corresponding changes in ANR and TCI are defined as
\begin{equation}
\Delta\mathrm{ANR}^{k}
=
\mathrm{ANR}(\boldsymbol{z}')
-
\mathrm{ANR}(\boldsymbol{z}^{*}),
\end{equation}
and
\begin{equation}
\Delta\mathrm{TCI}^{k}
=
\mathrm{TCI}(\boldsymbol{z}')
-
\mathrm{TCI}(\boldsymbol{z}^{*}).
\end{equation}
At the optimal sizing, the resulting changes satisfy
\begin{equation}
\Delta\mathrm{ANR}^{k}
-
\lambda^{*}
\Delta\mathrm{TCI}^{k}
\leq
\varepsilon,
\end{equation}
where $\varepsilon$ denotes the optimality tolerance of the parameterized MILP subproblem. A feasible increase in the maximum hourly production rate would improve the ANROI only if its incremental ANR exceeds the return required on the incremental TCI at the optimal ANROI, i.e.,
\begin{equation}
\Delta\mathrm{ANR}^{k}
>
\lambda^{*}
\Delta\mathrm{TCI}^{k}.
\end{equation}
At the optimum, no feasible increase satisfies this condition beyond the prescribed tolerance.
\par\noindent\textbf{Proof.}\\
A feasible increase in the maximum hourly production rate of unit $k$ improves the ANROI only if
\begin{equation}
\frac{\mathrm{ANR}(\boldsymbol{z}')}
{\mathrm{TCI}(\boldsymbol{z}')}
>
\frac{\mathrm{ANR}(\boldsymbol{z}^{*})}
{\mathrm{TCI}(\boldsymbol{z}^{*})}
=
\lambda^{*}.
\end{equation}
Since $\mathrm{TCI}(\boldsymbol{z}')>0$, this condition is equivalent to
\begin{equation}
\mathrm{ANR}(\boldsymbol{z}')
-
\lambda^{*}\mathrm{TCI}(\boldsymbol{z}')
>
0.
\end{equation}
Using
\begin{equation}
\mathrm{ANR}(\boldsymbol{z}^{*})
-
\lambda^{*}\mathrm{TCI}(\boldsymbol{z}^{*})
=
0,
\end{equation}
the above condition can be rewritten as
\begin{equation}
\Delta\mathrm{ANR}^{k}
-
\lambda^{*}\Delta\mathrm{TCI}^{k}
>
0.
\end{equation}
At convergence, $\boldsymbol{z}^{*}$ solves the parameterized MILP subproblem:
\[
\boldsymbol{z}^{*}
\in
\arg\max_{\boldsymbol{z}\in\mathcal{R}}
\left\{
\mathrm{ANR}(\boldsymbol{z})
-
\lambda^{*}\mathrm{TCI}(\boldsymbol{z})
\right\}.
\]
Therefore, for any feasible solution $\boldsymbol{z}'$,
\[
\mathrm{ANR}(\boldsymbol{z}')
-
\lambda^{*}\mathrm{TCI}(\boldsymbol{z}')
\leq
\mathrm{ANR}(\boldsymbol{z}^{*})
-
\lambda^{*}\mathrm{TCI}(\boldsymbol{z}^{*})
+
\varepsilon.
\]
Rearranging yields
\[
\Delta\mathrm{ANR}^{k}
-
\lambda^{*}
\Delta\mathrm{TCI}^{k}
\leq
\varepsilon.
\]
Hence, an increase in the maximum hourly production rate can improve the ANROI only when
\[
\Delta\mathrm{ANR}^{k}
>
\lambda^{*}\Delta\mathrm{TCI}^{k}.
\]
At the optimum, no feasible increase can satisfy this condition beyond the tolerance $\varepsilon$. For the original MILFP model, the economic value of hourly production rate flexibility is therefore assessed by comparing finite changes between feasible sizing solutions. This completes the proof.
\end{lemma}
\subsection{Evaluation metrics}
The optimal sizing results are further evaluated using the IRR, DPP, RECR, OSREUR.
\\
1) Internal rate of return \par
The IRR is commonly used to evaluate the profitability of engineering projects and is defined as the discount rate at which the project net present value (NPV) becomes zero \cite{patrick2016internal}. For a fixed project lifetime, ANROI and IRR have a strictly monotonic relationship \cite{yu2025novel}; therefore, maximizing ANROI is consistent with maximizing IRR. After the optimal ANROI is obtained, the corresponding IRR is further calculated as follows:
\begin{equation}
\mathrm{IRR}
=
\mathrm{CRF}^{-1}
\left[
\mathrm{ANROI}
+
\mathrm{CRF}(r_0,Y)
\right],
\end{equation}
$r_0$ is the benchmark discount rate;
$Y$ is the project lifetime;
and $\mathrm{CRF}^{-1}(\cdot)$ is the inverse function of the capital recovery factor.
\\
2) Discounted payback period\par
The DPP measures the time required to recover the initial investment based on discounted cash flows.
\begin{equation}
\mathrm{TCI}
=
\mathrm{NCF}
\frac{
1-(1+r_0)^{-\mathrm{DPP}}
}{
r_0
},
\end{equation}
The annual net cash flow is
\begin{equation}
\mathrm{NCF}
=
\mathrm{ANR}
+
AC_{\mathrm{Inv}}.
\end{equation}
The DPP can be calculated as follows:
\begin{equation}
\mathrm{DPP}
=
\frac{
\ln(\mathrm{NCF})
-
\ln(\mathrm{NCF}-r_0\mathrm{TCI})
}{
\ln(1+r_0)
}.
\end{equation}
3) Renewable energy curtailment rate\par
The RECR is defined as the ratio of curtailed RES to the total available RES generation:
\begin{equation}
\psi_{\mathrm{Curt}}
=
\frac{
\displaystyle
\sum_{t\in{T}}
P_{\mathrm{curt}}^{t}
}{
\displaystyle
\sum_{t\in{T}}
\left(
P_{\mathrm{W}}^{t}
+
P_{\mathrm{S}}^{t}
\right)
}
\times 100\%.
\end{equation}
4) Net on-site renewable energy utilization rate\par
The OSREUR is defined as the ratio of net on-site RES electricity consumed within the \ce{H2}-DRI-EAF-MeOH system to the total available on-site RES generation:
\begin{equation}
\psi_{\mathrm{RES,U}}
=
\frac{
\displaystyle
\sum_{t\in{T}}
\left(
P_{\mathrm{RES}}^{t}
-
P_{\mathrm{Sell}}^{t}
\right)
}{
\displaystyle
\sum_{t\in{T}}
\left(
P_{\mathrm{W}}^{t}
+
P_{\mathrm{S}}^{t}
\right)
}
\times 100\%.
\end{equation}
$\psi_{\mathrm{Curt}}$ is the RECR;
and $\psi_{\mathrm{RES,U}}$ is the OSREUR.
\section{Case study}
\label{sec4}
\subsection{Set up}
\label{sec4.1}
To validate the proposed model, the optimization framework was implemented in Python and solved using Gurobi. The annual wind and PV capacity-factor profiles are shown in Fig. 2, while the real-time electricity price and monthly capacity price profiles are shown in Fig. 3. The main operating and investment parameters are listed in Tables~\ref{tab:operation_parameters} and~\ref{tab:investment_parameters}, respectively.
\begin{table}[t]
\centering
\caption{Main operating parameters}
\label{tab:operation_parameters}
\small
\renewcommand{\arraystretch}{1.15}
\begin{tabular}{@{}cccc@{}}
\toprule
Parameters & Value & Parameters & Value \\
\midrule
$\psi_{\mathrm{Total,HDRI}}$ & 0.2211 
& $\psi_{\mathrm{CHDRI}}$ & 0.0521 \\

$\psi_{\mathrm{Tth,DRI}}$ & 0.92 
& $\psi_{\mathrm{ftg}}$ & 0.31 \\

$\psi_{\mathrm{whb}}$ & 0.12 
& $T_{\mathrm{AE}}$ & 1.27 h \\

$T_{\mathrm{SF,su}}/T_{\mathrm{SF,dn}}$ & 8/8 h 
& $\psi_{\mathrm{Eexp}}$ & 0.335 \\

$\psi_{\mathrm{ERcomp}}$ & 0.022 
& $\psi_{\mathrm{ECcomp}}$ & 0.278 \\

$\psi_{\mathrm{H2M}}$ & 0.1875 
& $\psi_{\mathrm{C2M}}$ & 1.375 \\
\bottomrule
\end{tabular}
\end{table}
\begin{table}[t]
\centering
\caption{Main investment and economic parameters}
\label{tab:investment_parameters}
\small
\renewcommand{\arraystretch}{1.15}
\begin{tabular}{@{}cccc@{}}
\toprule
Parameters & Value & Parameters & Value \\
\midrule
$r_0$ & 0.08 
& $Y$ & 20 \\
$\psi_{\mathrm{St,price}}$ & 720 
& $\psi_{\mathrm{DRI,price}}$ & 600 \\
$\psi_{\mathrm{MeOH,price}}$ & 600 
& $\psi_{\mathrm{He,price}}$ & 40 \\
$\psi_{\mathrm{Cb,price}}$ & 35 
& $\psi_{\mathrm{Deg}}$ & 4 \\
$\psi_{\mathrm{OD,price}}$ & 200 
& $\psi_{\mathrm{Sp,price}}$ & 400 \\
$I_{\mathrm{W}}$ & 650000 
& $I_{\mathrm{S}}$ & 550000 \\
$I_{\mathrm{AE}}$ & 300000 
& $I_{\mathrm{BS}}$ & 200000 \\
$I_{\mathrm{HS}}$ & 1000000 
& $I_{\mathrm{Lts}}$ & 200000 \\
$I_{\mathrm{CS}}$ & 500000 
& $I_{\mathrm{DS}}$ & 1000 \\
$C_{\mathrm{SF,max}}$ & 1500 
& $I_{\mathrm{SF,inv}}$ & 240 \\
$C_{\mathrm{EAF,max}}$ & 1500 
& $I_{\mathrm{EAF,inv}}$ & 240 \\
$C_{\mathrm{MSR,max}}$ & 1200 
& $I_{\mathrm{MSR,inv}}$ & 120 \\
$I_{\mathrm{Comp,inv}}$ & 500000 
& $I_{\mathrm{Exp,inv}}$ & 500000 \\
$I_{\mathrm{Leh,inv}}$ & 1000000 
& $I_{\mathrm{Heh,inv}}$ & 40000000 \\
\bottomrule
\end{tabular}
\end{table}
\begin{figure}[ht]
\centering
\includegraphics[width=0.7\textwidth]{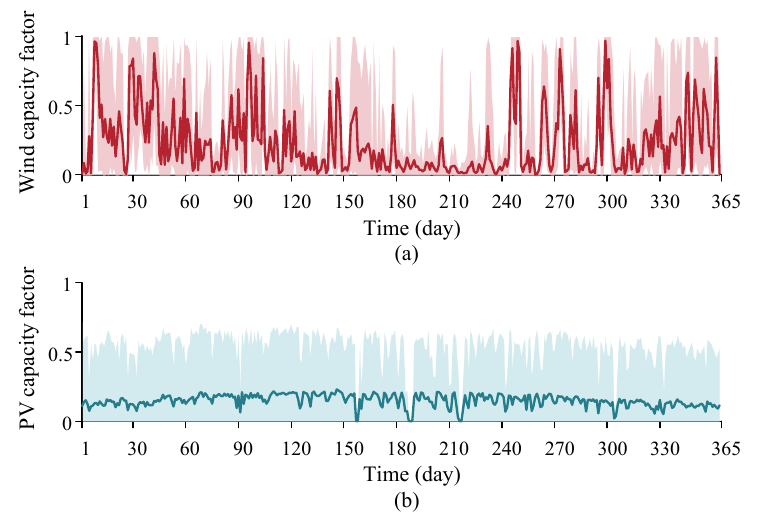}
\setlength{\abovecaptionskip}{-0.1cm}  
\setlength{\belowcaptionskip}{-0.1cm} 
\caption{Annual capacity-factor profiles: (a) Wind; (b) PV}
\captionsetup{justification=centering}
\vspace{-0.2cm}
\end{figure}
\begin{figure}[ht]
\centering
\includegraphics[width=0.7\textwidth]{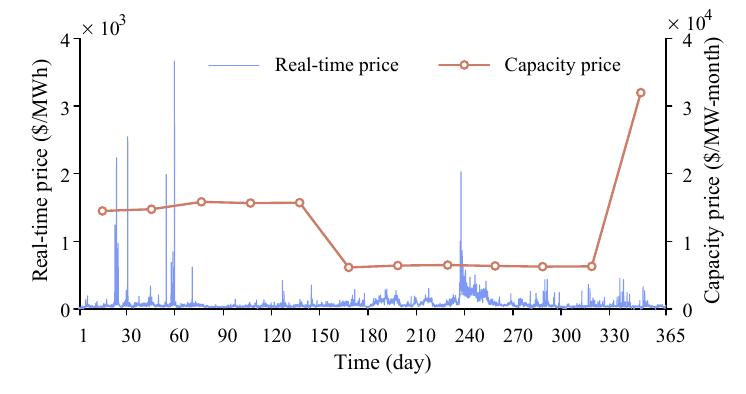}
\setlength{\abovecaptionskip}{-0.1cm}  
\setlength{\belowcaptionskip}{-0.1cm} 
\caption{Annual electricity price profiles: (a) real-time electricity price; (b) monthly capacity electricity price}
\captionsetup{justification=centering}
\vspace{-0.4cm}
\end{figure}
\subsection{On-grid optimal sizing and system performance}
Under the on-grid scenario, the ANROI-based optimal sizing solution outperforms the ANR benchmark in both economic performance and RES utilization. The IRR increases from 18.93\% to 20.67\%, corresponding to a relative increase of 9.20\%, while the DPP decreases from 6.84 to 6.17 years, corresponding to a relative reduction of 9.80\%. Meanwhile, the RECR decreases from 3.38\% to 0.01\%, and the OSREUR increases from 76.62\% to 83.03\%. These results show that the ANROI objective improves investment return while enhancing on-site RES utilization. \par
The optimized sizing includes 1317.37 MW of wind power, 1454.90 MW of PV, 610.56 MW of AE, 279.41 t-\ce{H2} of HS, and 3856.45 t-DRI of DS. The SF and EAF are sized at $1.0242\times10^{6}$ and $1.2000\times10^{6}$ t/year, respectively, with maximum hourly production rate limits of 157.97 and 300.00 t/h. Additionally, BS, CS, and MSR are not installed, indicating that grid trading substitutes for the short-term balancing role of BS, while the \ce{CO2}-to-MeOH pathway remains economically uncompetitive under the current parameter setting. \par
Under the on-grid scenario, the SF and EAF adjust their production rates to match temporal variations in RES generation. Fig. 4 presents the annual production profiles of DRI and crude steel. Both production profiles decline simultaneously around days 180--240, consistent with the low RES capacity factors, as shown in Fig. 2, and consequently low RES power output. This indicates that SF and EAF production remains largely driven by on-site RES generation. \par
\begin{figure}[ht]
\centering
\includegraphics[width=0.7\textwidth]{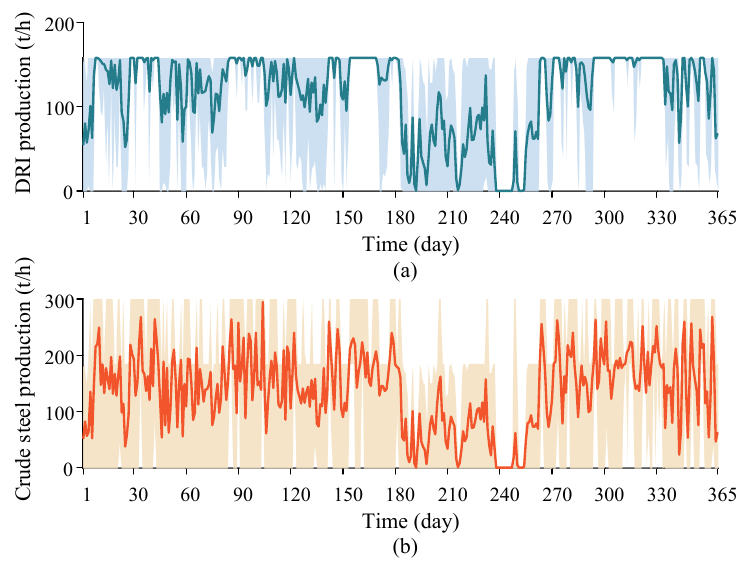}
\setlength{\abovecaptionskip}{-0.1cm}  
\setlength{\belowcaptionskip}{-0.1cm} 
\caption{Annual production profiles of the SF and EAF under the on-grid scenario: (a) DRI; (b) crude steel. Solid lines represent daily averages, and shaded areas indicate the daily minimum–maximum ranges.}
\captionsetup{justification=centering}
\vspace{-0.4cm}
\end{figure}
Meanwhile, HS and DS provide production side flexibility by balancing hydrogen supply and demand and buffering DRI flows, respectively. As shown in Fig. 5, the HS SOC fluctuates frequently over the year, indicating frequent charging and discharging to buffer short-term mismatches. In contrast, the DS exhibits sustained inventory accumulation and depletion, indicating longer-term buffering between DRI production in the SF and DRI consumption by the EAF.\par
\begin{figure}[ht]
\centering
\includegraphics[width=0.7\textwidth]{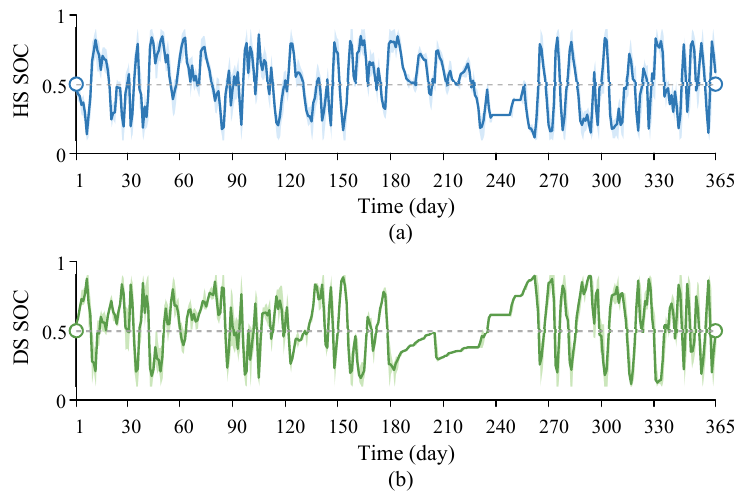}
\setlength{\abovecaptionskip}{-0.1cm}  
\setlength{\belowcaptionskip}{-0.1cm} 
\caption{Annual SOC profiles of storage units under the on-grid scenario: (a) HS; (b) DS.}
\captionsetup{justification=centering}
\vspace{-0.4cm}
\end{figure}
Additionally, the capacity electricity price suppresses monthly peak electricity purchases, while the real-time electricity price shapes electricity sales. As shown in Fig. 6, electricity purchases exhibit distinct stepwise plateaus across different months, whereas electricity sales fluctuate frequently over the year.
\begin{figure}[ht]
\centering
\includegraphics[width=0.7\textwidth]{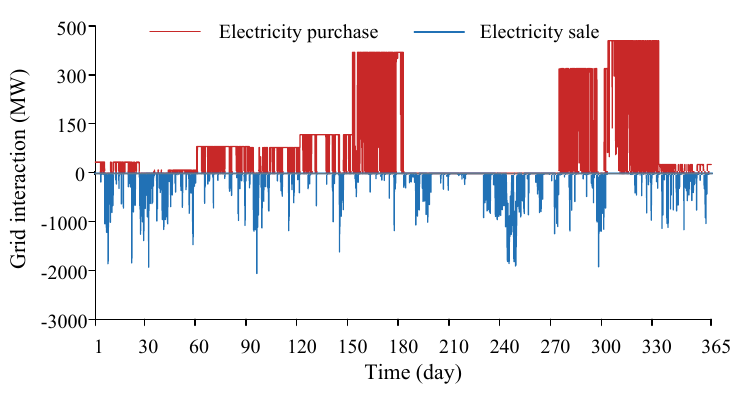}
\setlength{\abovecaptionskip}{-0.1cm}  
\setlength{\belowcaptionskip}{-0.1cm} 
\caption{Annual grid electricity purchase and sale profiles under the on-grid scenario.}
\captionsetup{justification=centering}
\vspace{-0.4cm}
\end{figure}
\subsection{Off-grid optimal sizing and system performance}
Under the off-grid scenario, the ANROI-based optimal sizing solution marginally outperforms the ANR benchmark in both economic performance and RES utilization. The IRR increases from 16.90\% to 17.08\%, while the DPP decreases from 7.83 to 7.73 years. Meanwhile, the RECR decreases from 5.00\% to 4.76\%, and the OSREUR increases from 95.00\% to 95.24\%. These results indicate that the ANROI objective improves investment return and RES utilization, although the improvement is less pronounced than that under the on-grid scenario. \par
The optimized sizing comprises 1193.60 MW of wind power, 1745.97 MW of PV, 712.27 MW of AE, 334.68 t-\ce{H2} of HS, and 6428.61 t-DRI of DS. The annual production capacities of the SF and EAF are $1.0242\times10^{6}$ and $1.2000\times10^{6}$ t/year, respectively, with maximum hourly production-rate limits of 189.21 and 300.00 t/h. A small BS capacity of 3.69 MWh is installed, whereas CS and MSR are not installed, indicating that BS provides only a marginal balancing contribution and that the \ce{CO2}-to-MeOH pathway remains economically uncompetitive under the current parameter setting. \par
Compared with the on-grid optimal sizing, the off-grid system requires larger RES, AE, HS, and DS capacities to maintain internal energy and material balance without grid support. Grid interaction reduces these capacity requirements while providing additional revenue through electricity trading, thereby resulting in a higher IRR and a shorter DPP. \par
\begin{figure}[ht]
\centering
\includegraphics[width=0.7\textwidth]{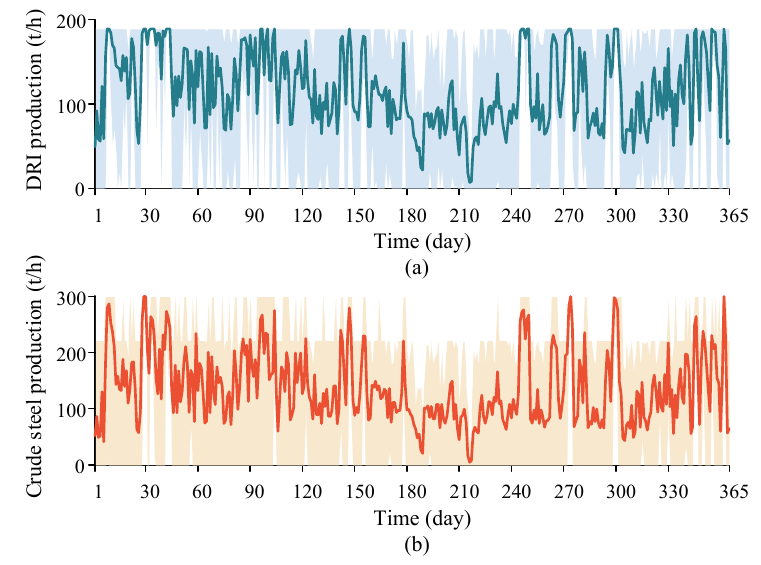}
\setlength{\abovecaptionskip}{-0.1cm}  
\setlength{\belowcaptionskip}{-0.1cm} 
\caption{Annual production profiles of the SF and EAF under the off-grid scenario: (a) DRI; (b) crude steel. }
\captionsetup{justification=centering}
\vspace{-0.4cm}
\end{figure}
\begin{figure}[ht]
\centering
\includegraphics[width=0.7\textwidth]{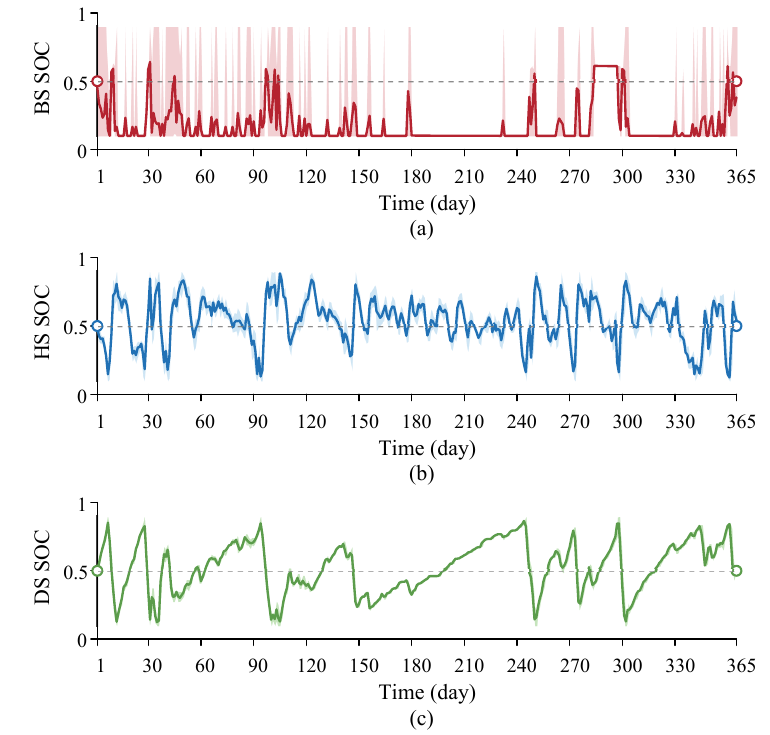}
\setlength{\abovecaptionskip}{-0.1cm}  
\setlength{\belowcaptionskip}{-0.1cm} 
\caption{Annual SOC profiles of storage units under the off-grid scenario: (a) BS; (b) HS; (c) DS.}
\captionsetup{justification=centering}
\vspace{-0.4cm}
\end{figure}
Additionally, Figs. 7 and 8 show that DRI and crude steel production vary with RES generation, while HS and DS provide the principal buffering and the BS remains marginally utilized. This indicates that internal mass and energy balancing primarily relies on production side flexibility and the buffering of HS and DS.
\subsection{System value of hourly production rate flexibility}
Hourly production rate flexibility expands the feasible operating range of the SF and EAF, thereby improving economic performance and RES utilization. This section evaluates its system value under both on-grid and off-grid scenarios.
\subsubsection{On-grid scenario}
Under the on-grid scenario, greater hourly production rate flexibility increases the system IRR, but with diminishing marginal gains. As shown in Fig. 9, with the annual production capacities held constant, fewer equivalent annual operating hours correspond to higher maximum hourly production rates of the SF and EAF and, therefore, greater production-rate flexibility. The IRR increases sharply as the equivalent annual operating hours decrease from 8000 to 7000 h, whereas the curve gradually flattens with further reductions, indicating diminishing marginal returns from additional hourly production-rate flexibility. Overall, reducing the equivalent annual operating hours from 8000 to 4000 h increases the system IRR from 19.68\% to 20.67\%, a relative increase of 5.03\%. \par
\begin{figure}[ht]
\centering
\includegraphics[width=0.7\textwidth]{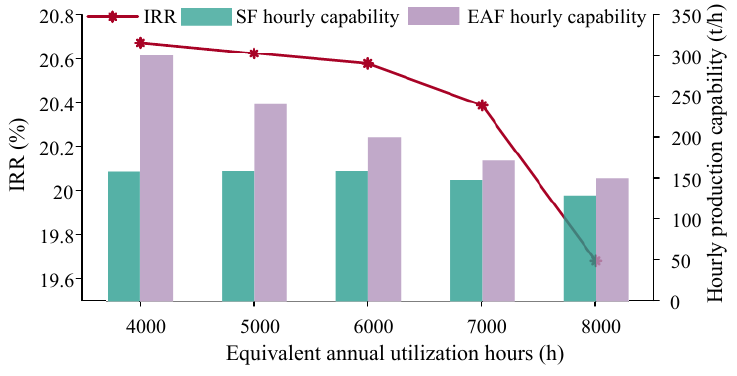}
\setlength{\abovecaptionskip}{-0.1cm}  
\setlength{\belowcaptionskip}{-0.1cm} 
\caption{Maximum hourly production rate flexibility and system IRR under the on-grid scenario}
\captionsetup{justification=centering}
\vspace{-0.2cm}
\end{figure}
Table 3 further compares the optimal sizing results at 8000 and 4000 equivalent annual operating hours. The maximum hourly production rates of the SF and EAF increase from 129.00 and 150.00 t/h to 157.99 and 300.00 t/h, respectively, while the BS capacity decreases from 113.56 MWh to 0. Additionally, the HS and DS capacities increase from 228.17 t-\ce{H2} and 475.94 t-DRI to 279.46 t-\ce{H2} and 3856.52 t-DRI, respectively. These results indicate that greater hourly production rate flexibility improves the temporal alignment of production with RES generation, thereby reducing the BS requirement while increasing the storage requirements for hydrogen and DRI.
\begin{table}[htbp]
    \centering
    \caption{On-grid optimal sizing at different hourly production-rate flexibility levels}
    \resizebox{\linewidth}{!}{
    \begin{tabular}{cccccc}
        \toprule
        \begin{tabular}[c]{@{}c@{}}
        Equivalent annual\\
        operating hours (h)
        \end{tabular}
        &
        \begin{tabular}[c]{@{}c@{}}
        SF production rate\\
        (t/h)
        \end{tabular}
        &
        \begin{tabular}[c]{@{}c@{}}
        EAF production rate\\
        (t/h)
        \end{tabular}
        &
        \begin{tabular}[c]{@{}c@{}}
        BS\\
        (MWh)
        \end{tabular}
        &
        \begin{tabular}[c]{@{}c@{}}
        HS\\
        (t-\ce{H2})
        \end{tabular}
        &
        \begin{tabular}[c]{@{}c@{}}
        DS\\
        (t-DRI)
        \end{tabular}
        \\
        \midrule
        4000 & 157.99 & 300.00 & 0      & 279.46 & 3856.52 \\
        8000 & 129.00 & 150.00 & 113.56 & 228.17 & 475.94  \\
        \bottomrule
    \end{tabular}
    }
\end{table}
\subsubsection{Off-grid scenario}
Without grid interaction, the economic value of hourly production rate flexibility becomes substantially more pronounced. As shown in Fig. 10, reducing the equivalent annual operating hours from 8000 to 4000 h increases the system IRR from 12.64\% to 17.08\%, a relative increase of 35.22\%, substantially exceeding the 5.03\% obtained under the on-grid scenario. \par
\begin{table}[htbp]
    \centering
    \caption{Off-grid optimal sizing at different hourly production-rate flexibility levels}
    \resizebox{\linewidth}{!}{
    \begin{tabular}{ccccccc}
        \toprule
        \begin{tabular}[c]{@{}c@{}}
        Equivalent annual\\
        operating hours (h)
        \end{tabular}
        &
        \begin{tabular}[c]{@{}c@{}}
        SF production rate\\
        (t/h)
        \end{tabular}
        &
        \begin{tabular}[c]{@{}c@{}}
        EAF production rate\\
        (t/h)
        \end{tabular}
        &
        \begin{tabular}[c]{@{}c@{}}
        BS\\
        (MWh)
        \end{tabular}
        &
        \begin{tabular}[c]{@{}c@{}}
        HS\\
        (t-\ce{H2})
        \end{tabular}
        &
        \begin{tabular}[c]{@{}c@{}}
        DS\\
        (t-DRI)
        \end{tabular}
        \\
        \midrule
        4000 & 189.21 & 300.00 & 3.69   & 334.68 & 6428.61 \\
        8000 & 73.47  & 75.89  & 800.00 & 129.95 & 8283.56 \\
        \bottomrule
    \end{tabular}
    }
\end{table}
\begin{figure}[ht]
\centering
\includegraphics[width=0.7\textwidth]{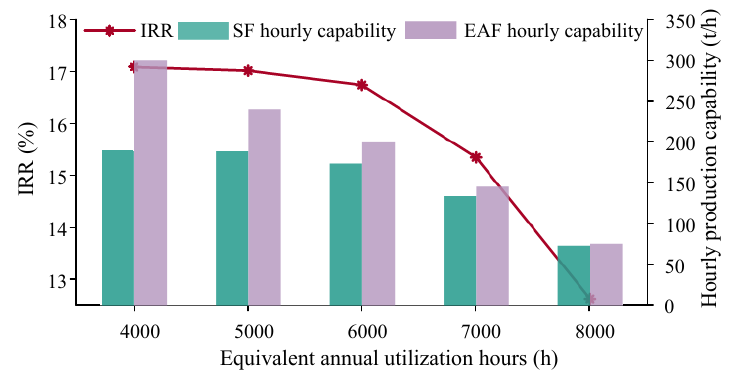}
\setlength{\abovecaptionskip}{-0.1cm}  
\setlength{\belowcaptionskip}{-0.1cm} 
\caption{Maximum hourly production rate flexibility and system IRR under the off-grid scenario}
\captionsetup{justification=centering}
\vspace{-0.2cm}
\end{figure}
Table 4 further compares the off-grid sizing results at 8000 and 4000 equivalent annual operating hours. The maximum hourly production rates of the SF and EAF increase from 73.47 and 75.89 t/h to 189.21 and 300.00 t/h, respectively. Meanwhile, the BS capacity decreases sharply from 800.00 to 3.69 MWh, whereas the HS increases from 129.95 to 334.68 t-\ce{H2} and the DS decreases from 8283.56 to 6428.61 t-DRI. These changes indicate that greater hourly production rate flexibility substantially reduces reliance on BS for off-grid power balancing, while increasing the HS capacity and reducing the DS capacity.
\subsection{Optimal sizing under off-grid zero-carbon scenario}
The off-grid zero-carbon constraint substantially reduces the economic performance of the integrated \ce{H2}-DRI-EAF-MeOH system. The system achieves an IRR of 14.99\%, a DPP of 9.04 years, an RECR of 5.89\%, and an OSREUR of 94.11\%. Compared with the off-grid scenario without the zero-carbon constraint, the IRR decreases from 17.08\% to 14.99\%, corresponding to a relative reduction of 12.24\%. At the MeOH production scale of $9.59\times10^{4}$ t/year, the revenues from MeOH sales and carbon allowances remain insufficient to offset the investment and operating costs of the CCU-MeOH pathway. A further comparison with the on-grid scenario shows that the IRR decreases from 20.67\% to 14.99\%, corresponding to a relative reduction of 27.48\%. The wider gap mainly reflects the loss of the balancing and revenue benefits of grid interaction.\par
The optimized sizing includes 1505.01 MW of wind power, 2000.00 MW of PV, 911.09 MW of AE, 362.14 t-\ce{H2} of HS, and 938.32 t-DRI of DS. The annual production capacities of the SF and EAF are $9.931\times10^{5}$ and $1.1626\times10^{6}$ t/year, respectively, with maximum hourly production-rate limits of 190.33 and 290.66 t/h. In addition, 186.91 t-\ce{CO2} of CS and an MSR with an annual MeOH production capacity of $9.59\times10^{4}$ t/year are installed. Compared with the off-grid scenario without the zero-carbon constraint, larger RES, AE, and HS capacities are required to satisfy the additional electricity and hydrogen demand introduced by MeOH synthesis while maintaining hydrogen supply for iron reduction. \par
Additionally, Fig. 11 presents the coordinated operation of ironmaking, steelmaking, and MeOH synthesis under the zero-carbon constraint. MeOH and crude steel production exhibit similar temporal patterns. The limited CS capacity of 186.91 t-\ce{CO2} and the storage profiles in Fig. 12 indicate that HS and DS provide the main intertemporal buffering, whereas CS only mitigates short-term mismatches between EAF \ce{CO2} capture and MeOH synthesis.
\begin{figure}[ht]
\centering
\includegraphics[width=0.7\textwidth]{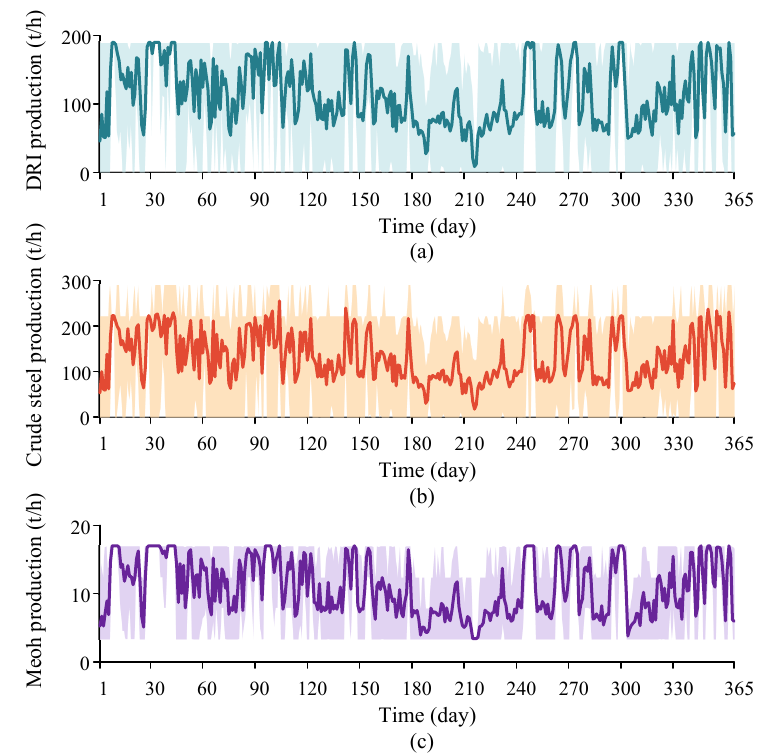}
\setlength{\abovecaptionskip}{-0.1cm}  
\setlength{\belowcaptionskip}{-0.1cm} 
\caption{Annual production profiles of the SF, EAF, and MSR under the off-grid zero-carbon emission scenario: (a) DRI; (b) crude steel; (c) MeOH.}
\captionsetup{justification=centering}
\vspace{-0.2cm}
\end{figure}
\begin{figure}[ht]
\centering
\includegraphics[width=0.7\textwidth]{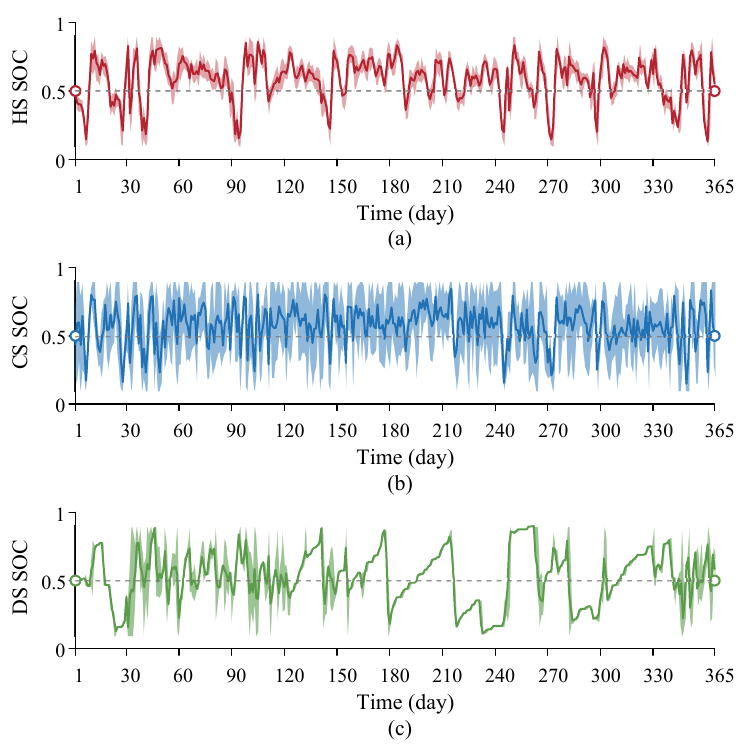}
\setlength{\abovecaptionskip}{-0.1cm}  
\setlength{\belowcaptionskip}{-0.1cm} 
\caption{Annual SOC profiles of storage units under the off-grid zero-carbon emission scenario: (a) HS; (b) CS; (c) DS.}
\captionsetup{justification=centering}
\vspace{-0.2cm}
\end{figure}
\section{Conclusion}
\label{sec5}
This study develops a fractional programming-based optimal sizing model for an integrated \ce{H2}-DRI-EAF-MeOH system to maximize investment return. The techno-economic performance of the system is analyzed under the on-grid, off-grid, and off-grid zero-carbon scenarios. The main conclusions are as follows:\par
1) The ANROI-based sizing improves system investment return under both the on-grid and off-grid scenarios. Compared with the ANR benchmark, the system IRR increases from 18.93\% to 20.67\% under the on-grid scenario, corresponding to a relative increase of 9.20\%;  under the off-grid scenario, it increases from 16.90\% to 17.08\%, corresponding to a relative increase of 1.07\%. \par
2) Greater hourly production-rate flexibility improves economic performance and substantially reduces reliance on BS for short-term power balancing. As the equivalent annual operating hours decrease from 8000 to 4000 h, the on-grid IRR increases from 19.68\% to 20.67\%, corresponding to a relative increase of 5.03\%, while the BS capacity decreases by 113.56 MWh. Under the off-grid scenario, the IRR increases from 12.64\% to 17.08\%, corresponding to a relative increase of 35.22\%, while the BS capacity decreases by 796.31 MWh. \par 
3) The off-grid zero-carbon constraint requires the deployment of the CCU-MeOH pathway but lowers system investment return. Compared with the off-grid scenario, the system IRR decreases from 17.08\% to 14.99\%, indicating that MeOH and carbon-allowance revenues are insufficient to offset the investment and operational costs of the CCU-MeOH pathway. \par
In the future, large-scale integration between steel plants and chemical industrial parks will be investigated by incorporating multi-product \ce{CO2} utilization pathways and analyzing their techno-economic performance.\par
\section*{Credit authorship contribution statement}
\noindent Qiang Ji: Conceptualization, Methodology, Software, Writing. Fashun Shi: Editing, Supervision, Conceptualization. Lin Cheng: Conceptualization, Methodology. Yeye Xie: Editing, Supervision. Yufei Xi: Editing, Supervision. Zhen Dai: Conceptualization. Xun Wang: Supervision, Conceptualization.\par
\vspace{0.3cm}
\noindent \textbf{Declaration of competing interest}\par
\noindent The authors declare that they have no known competing financial interests or personal relationships that could have appeared to influence the work reported in this paper.
\vspace{0.3cm}
\par
\noindent \textbf{Acknowledgements}
\noindent This work was supported by Smart-Grid National Science and Technology Major Project (No. 2025ZD0806300).

\bibliographystyle{elsarticle-num}
\bibliography{AE}

\end{document}